\documentclass[nohyperref,11pt]{article}
\usepackage{algorithm}
\usepackage[noend]{algorithmic}
\usepackage{microtype}
\usepackage{subfigure}
\usepackage{booktabs}
\usepackage{graphicx}
\usepackage{xcolor}
\usepackage{amsmath}
\usepackage{amssymb}
\usepackage{amsthm}
\usepackage{mathtools}
\usepackage{wrapfig}
\usepackage[margin=0.8in,footskip=0.5in]{geometry}

\usepackage{hyperref}
\usepackage{xspace}
\hypersetup{
     colorlinks = true,
     linkcolor = blue,
     citecolor = blue,
     filecolor = blue,
     urlcolor = blue
     }
\usepackage{cleveref}

\usepackage{comment}
\usepackage{enumitem}
\usepackage{quoting}
\quotingsetup{vskip=0pt}

\allowdisplaybreaks

\newcommand{\para}[1]{\medskip \par \noindent {\bf #1}}

\newtheorem{lemma}{Lemma}
\newtheorem{theorem}[lemma]{Theorem}

\newtheorem{definition}[lemma]{Definition}

\newtheorem{proposition}[lemma]{Proposition}

\newcommand{\bw}{\mathbf{w}}

\newcommand{\ty}{\tilde{y}}

\newcommand{\eps}{\varepsilon}
\newcommand{\E}{\mathbb{E}}
\newcommand{\Z}{\mathbb{Z}}
\newcommand{\A}{\mathbb{A}}
\newcommand{\cD}{\mathcal{D}}
\newcommand{\cL}{\mathcal{L}}
\newcommand{\cR}{\mathcal{R}}
\newcommand{\cI}{\mathcal{I}}
\newcommand{\hf}{\hat{f}}
\newcommand{\Ileft}{I_{\text{left}}}
\newcommand{\Iright}{I_{\text{right}}}

\DeclareMathOperator{\Thre}{Thresh}
\DeclareMathOperator{\point}{point}
\DeclareMathOperator{\pol}{pol}

\newcommand{\bD}{x}%
\newcommand{\cA}{\mathcal{A}}

\newcommand{\cH}{\mathcal{H}}

\newcommand{\n}{n}

\newcommand{\p}{p}
\newcommand{\s}{s}

\newcommand{\w}{w}

\newcommand{\y}{y}

\newcommand{\dec}{\text{dec}}
\newcommand{\clip}{\text{clip}}
\newcommand{\hs}{\text{halfspace}}
\newcommand{\z}{z}
\newcommand{\F}{\mathcal{F}}
\newcommand{\C}{\mathcal{C}}
\newcommand{\N}{\mathbb{N}}
\newcommand{\R}{\mathbb{R}}
\newcommand{\bone}{\mathbf{1}}
\newcommand{\PDP}{$\nabla_0$DP\xspace}

\newcommand{\ex}[2]{{\ifx&#1& \mathbb{E} \else
\underset{#1}{\mathbb{E}} \fi \left[#2\right]}}
\newcommand{\pr}[2]{{\ifx&#1& \mathbb{P} \else
\underset{#1}{\mathbb{P}} \fi \left[#2\right]}}
\newcommand{\var}[2]{{\ifx&#1& \mathsf{Var} \else
\underset{#1}{\mathsf{Var}} \fi \left[#2\right]}}
\newcommand{\dr}[3]{\mathrm{D}_{#1}\left(#2\middle\|#3\right)}

\iftrue
\usepackage[style=alphabetic,backend=bibtex,maxalphanames=10,maxbibnames=20,maxcitenames=10,giveninits=true,doi=false,url=true,backref=true]{biblatex}
\newcommand*{\citet}[1]{\AtNextCite{\AtEachCitekey{\defcounter{maxnames}{2}}}
\textcite{#1}}

\newcommand*{\citep}[1]{\cite{#1}}

\fi

\DeclareMathOperator*{\argmin}{arg\,min}
\DeclareMathOperator*{\sgn}{sign}
\DeclareMathOperator*{\Lap}{Lap}
\DeclareMathOperator*{\err}{err}
\DeclareMathOperator*{\poly}{poly}

\iftrue
\newcommand{\pasin}[1]{\textcolor{red}{[Pasin: #1]}}
\newcommand{\thomas}[1]{\textcolor{red}{[Thomas: #1]}}
\newcommand{\badih}[1]{\textcolor{red}{[Badih: #1]}}
\newcommand{\todo}[1]{\textcolor{red}{[Todo: #1]}}
\else
\newcommand{\pasin}[1]{}
\newcommand{\thomas}[1]{}
\newcommand{\badih}[1]{}
\newcommand{\todo}[1]{}
\fi

\title{Algorithms with More Granular Differential Privacy Guarantees}

\author{
Badih Ghazi
\and
Ravi Kumar
\and
Pasin Manurangsi
\and
Thomas Steinke
}
\date{}

\begin{document}

\maketitle
\footnotetext{Alphabetical order. Google Research. \texttt{badihghazi@google.com, ravi.k53@gmail.com, pasin@google.com, partialdp@thomas-steinke.net}}

\begin{abstract}
    Differential privacy is often applied with a privacy parameter that is larger than the theory suggests is ideal; various informal justifications for tolerating large privacy parameters have been proposed.
    In this work, we consider partial differential privacy (DP), which allows quantifying the privacy guarantee on a per-attribute basis. 
    In this framework, we study several basic data analysis and learning tasks, and design algorithms whose per-attribute privacy parameter is smaller that the best possible privacy parameter for the entire record of a person (i.e., all the attributes).
\end{abstract}

\newpage

\begin{small}
\tableofcontents
\end{small}

\newpage

\section{Introduction}

Differential Privacy (DP) \cite{DworkMNS06} provides a strict worst-case privacy guarantee --- even an adversary that knows the entire dataset except for one bit of information about one individual cannot learn that bit, even when the dataset and the bit in question are arbitrary. Since its inception, researchers have sought to relax the DP definition in order to permit better data analysis while still providing meaningful privacy guarantees \cite{DesfontainesP19}.

The only approach to relaxing the definition of DP that has gained widespread use --- albeit not acceptance --- is quantitative relaxation. That is, it is common to set the main privacy parameter (usually denoted by $\varepsilon$) to be larger than the theory allows us to easily interpret. 
More precisely, the privacy loss bound $\varepsilon$ is used to quantify the tolerable accuracy with which an adversary can learn the unknown bit. Theory would suggest that $\varepsilon \le 1$ provides a good privacy guarantee, and that the guarantee rapidly degrades if we further increase $\varepsilon$. The setting $\varepsilon=10$ permits a sensitive bit to be revealed with $99.995\%$ accuracy, if we are unlucky enough to be in a truly worst-case setting. Nevertheless, it is common to set $\varepsilon \ge 10$ \cite{desfontaines2020list,tang2017privacy,censusparams2021}. 
This raises the question: \emph{Can one still provide meaningful privacy guarantees even when the DP parameter $\varepsilon$ is large?}

Informally, we can justify large DP parameters by arguing that, for ``realistic'' adversaries, ``natural'' data, and ``nice'' algorithms, the ``real'' privacy guarantee is better than the worst-case guarantee summarized by the single parameter $\varepsilon$. (And, of course, we can also hope that better privacy parameters could be obtained by a more careful analysis of the algorithm.)

In this paper, we seek to understand the above intuitive justification for tolerating large DP parameters. This requires us to formalize what constitutes a ``nice'' algorithm, which we do with \emph{partial DP} --- a notion based on previous studies in the DP literature that provides a more granular accounting of the privacy loss parameter. In particular, this permits us to quantify a ``per-attribute $\varepsilon_0$'' in addition to the usual ``per-person $\varepsilon$.'' (See Section~\ref{sec:prelim} for the formal definition, and Section~\ref{sec:related-work} for an overview of the prior work.)  Setting $\varepsilon \ge 10$ may become more palatable if, e.g., we can simultaneously assert that each (sensitive) attribute has $\varepsilon_0 \le 1$. To interpret such a guarantee, we must also discuss what sort of adversaries and attacks we are and are not protected against.

The focus of our work is on designing and analyzing algorithms that establish a quantitative separation between the attainable per-person $\varepsilon$ and the per-attribute $\varepsilon_0$. 
The key message is that we demonstrate that, in many circumstances, we can say more about the privacy properties of an algorithm beyond what can be conveyed by the standard single-parameter definition of DP.

\subsection{Contributions}
We investigate a variety of fundamental data analysis and learning tasks through the lens of per-attribute partial DP. That is, we present several algorithms and analyze their fine-grained privacy properties.
More specifically, under the partial DP notion, we consider three data analysis tasks and obtain more granular bounds than under standard DP:

(i) We first study algorithms for answering general families of statistical queries (\Cref{sec:gen_alg_results}). We analyze the projection mechanism and a variant of the multiplicative weights exponential mechanism under per-attribute partial DP. These results show a \emph{separation} between the standard per-person $\varepsilon$ and the per-attribute $\varepsilon_0$ that scales polynomially with the dimension (i.e., the number of attributes in each person's record).
    
(ii) We next present a new algorithm for computing histograms (a.k.a., heavy hitters) that gives a per-attribute privacy parameter $\varepsilon_0$ that is \emph{exponentially smaller} than the standard per-person $\varepsilon$-DP parameter in terms of the number of attributes (\Cref{sec:heavy_hitters}). That is, if each person's record is $d$ bits, the error of our algorithm grows as $\frac{\log d}{\varepsilon_0}$, while the standard pure DP algorithm would have error $\frac{d}{\varepsilon}$. We also prove a near-matching lower bound.
    
Note that histograms are an important case study, as the standard algorithms for this closely resemble the kind of worst-case algorithms that we wish to rule out. E.g., if  we add Laplace noise to each count in a histogram to achieve $\varepsilon$-DP with $\varepsilon \ge 10$, then $\ge 99\%$ of counts would still round back to the exact value. Depending on the sparsity of the histogram, this would be a weak privacy guarantee in practice. Hence we design a partial DP algorithm that avoids this worst-case behaviour.
    
(iii)  Finally, we present an algorithm for robustly learning halfspaces under per-attribute partial DP (Section \ref{sec:heavy_hitters}). This has a per-attribute privacy parameter that does not grow with the dimension, as is the case for the standard per-person privacy parameter.  %

Our results are summarized in Table~\ref{tab:result summary}.
\setlength{\tabcolsep}{1pt}
\begin{table}
    \centering
    \begin{scriptsize}
    \begin{tabular}{|c|c|c|}
        \hline
         Task & Standard $\varepsilon$-DP / $\frac12\varepsilon^2$-zCDP & Per-attribute $\varepsilon_0$-$\nabla_0$CDP \\
         \hline
         $\mathcal{Q} = \{ \text{$k$-way marginals on $\{0,1\}^d$}\}$, average error $\alpha$ & $n = O\left(\min\left\{ \frac{\sqrt{d}}{\varepsilon\alpha^2}, \frac{\sqrt{|\mathcal{Q}|}}{\varepsilon\alpha}\right\}\right)$ & $n \!=\! O\!\left(\!\min\!\left\{ \frac{\sqrt{k}}{\varepsilon_0\alpha^2},\! \frac{\sqrt{|\mathcal{Q}| \cdot k/d}}{\varepsilon_0\alpha}\!\right\}\!\right)\!$\\
         \hline
         $\mathcal{Q} = \{ \text{$k$-way marginals on $\{0,1\}^d$}\}$, max$^*$ error $\alpha$ & $n = O\left(\frac{\sqrt{d} \cdot k \log d}{\varepsilon \alpha^2}\right)$ & $n = O\left(\frac{k^{3/2} \cdot \log d}{\varepsilon_0 \alpha^2}\right)$\\
         \hline
         Histogram / Heavy Hitters on $\mathcal{X}=\{0,1\}^d$, error $\alpha$ & $n = O\left(\frac{d}{\varepsilon\alpha}\right)$ / $n = O\left(\frac{\sqrt{d}}{\varepsilon\alpha}\right)$ & $n = O\left(\frac{\log d}{\varepsilon_0\alpha}\right)$ \\
         \hline
         $\gamma$-robust learning of halfspaces over $\{\pm 1\}^d$ & $n\!=\!O\!\left(\frac{d\log(1/\alpha)}{\alpha\varepsilon} \!+\! \frac{\log(1/\alpha)}{\varepsilon\alpha\gamma^2} \!+\! \frac{1}{\alpha^2\gamma^2}\right)\!$ & $n={O}\left(\frac{1}{\varepsilon_0^2 \alpha^2 \gamma^4} + \frac{\log(1/\gamma)}{\varepsilon_0^2 \alpha \gamma}\right)$\\
         \hline
    \end{tabular}
    \caption{Summary of sample complexities of our algorithmic results for per-attribute partial DP, compared to standard DP. See Section~\ref{sec:prelim} for the definition of $\frac12\varepsilon^2$-zCDP and $\varepsilon_0$-$\nabla_0$CDP.}
    \label{tab:result summary}
    \end{scriptsize}
\end{table}

\section{Formal Definitions \& Basic Properties}
\label{sec:prelim}

We briefly recall the definition of differential privacy (DP)~\cite{DworkMNS06,DworkKMMN06}.
\begin{definition}[DP]
    A randomized algorithm $M : \mathcal{X}^n \to \mathcal{Y}$ is \emph{$(\varepsilon,\delta)$-differentially private ($(\varepsilon,\delta)$-DP)} if, for all $x,x' \in \mathcal{X}^n$ differing on a single entry (i.e., $\exists i \in [n] ~ \forall j \in [n]\setminus\{i\} ~x_j=x'_j$) and all measurable $S \subset \mathcal{Y}$, we have $\pr{}{M(x)\in S} \le e^\varepsilon \cdot \pr{}{M(x')\in S}+\delta$.
\end{definition}
The setting with $\delta=0$ (abbreviated  $\varepsilon$-DP) is called \emph{pure} DP, while $\delta>0$ is called \emph{approximate} DP.  We also work with zero-concentrated DP (zCDP)~\cite{BunS16}, which is a refinement of the original definition of concentrated DP~\cite{DworkR16}. This is formulated via R\'enyi divergences \cite{Renyi61}:
\begin{definition}\label{defn:pdp}
    Let $P$ and $Q$ be probability distributions on $\Omega$ with a common $\sigma$-algebra.\footnote{Formally, we assume that $Q(S)=0 \implies P(S) = 0$ for all measurable $S \subset \Omega$. (If this assumption does not hold, define $\dr{\lambda}{P}{Q}=\infty$ for all $\lambda \in [1,\infty]$.) Let $P(x)/Q(x)$ denote the Radon--Nikodym derivative of $P$ with respect to $Q$ evaluated at $x \in \Omega$ so that $P(S) := \ex{X \gets P}{\mathbb{I}[X \in S]} = \ex{X \gets Q}{\mathbb{I}[X \in S] \cdot P(X)/Q(X)}$ for all measurable $S \subset \Omega$.} For $\lambda \in (1,\infty)$, define \(\dr{\lambda}{P}{Q} := \frac{1}{\lambda-1} \log \ex{X \gets P}{\left(\frac{P(x)}{Q(x)}\right)^{\lambda-1}}.\)
    We define $\dr{*}{P}{Q} := \sup_{\lambda \in (1,\infty)} \frac{1}{\lambda} \dr{\lambda}{P}{Q}$ and $\dr{\infty}{P}{Q} := \sup_{S \subset \Omega : P(S)>0} \log(P(S)/Q(S))$.
\end{definition}
\begin{definition}[CDP]
    A randomized algorithm $M : \mathcal{X}^n \to \mathcal{Y}$ is \emph{$\frac12\varepsilon^2$-zero-Concentrated DP ($\frac12\varepsilon^2$-zCDP)} if, for all $x,x' \in \mathcal{X}^n$ differing on a single entry, $\dr{*}{M(x)}{M(x')} \le \frac12\varepsilon^2$.
\end{definition}

\subsection{Partial Differential Privacy}

Partial DP is a natural extension of the standard DP definition; we replace the single parameter $\varepsilon$ with a function that measures the dissimilarity of two persons' records. Similar definitions have appeared in the literature before (Section~\ref{sec:related-work}).
\begin{definition}[Partial DP]\label{defn:ppdp}
Let $\varepsilon : \mathcal{X} \times \mathcal{X} \to \mathbb{R}$ be symmetric and non-negative (i.e., $\forall x,x'\in\mathcal{X},~\varepsilon(x,x')=\varepsilon(x',x)\ge 0$).\footnote{By group privacy, we can assume, without loss of generality, that $\varepsilon$ also satisfies the triangle inequality $\varepsilon(x,x'') \le \varepsilon(x,x')+\varepsilon(x',x'')$ for all $x,x',x''\in\mathcal{X}$.}
We say that a randomized algorithm $M : \mathcal{X}^n \to \mathcal{Y}$ is \emph{$\varepsilon$-partially DP ($\varepsilon$-$\nabla$DP)} if, for all inputs $x,x'\in \mathcal{X}^n$ differing only on a single entry and for all measurable $S \subset \mathcal{Y}$, 
\begin{equation}\pr{}{M(x) \in S} \le e^{\varepsilon(x_i,x'_i)} \cdot \pr{}{M(x') \in S},
\label{eq:ppdp}\end{equation} 
where  $i \in [n]$ is the index of the entry on which $x$ and $x'$ differ (i.e., $\forall i' \in [n]\setminus\{i\} ~ x_{i'} = x'_{i'}$).
\end{definition}
\begin{definition}[Per-Attribute Partial DP]
For $x,x'\in\mathcal{X} = \mathcal{X}_1 \times \cdots \times \mathcal{X}_d$, we denote the Hamming distance $\|x-x'\|_0 := |\{j \in [d] : x_j \ne x'_j\}|$. For $\varepsilon_0 \ge 0$, we define \emph{$\varepsilon_0$-per-attribute partial DP ($\varepsilon_0$-$\nabla_0$DP)} (which we also call \emph{per-attribute DP} or \emph{Hamming partial DP}) to be $\varepsilon$-$\nabla$DP with $\varepsilon(x,x') := \varepsilon_0 \cdot \|x-x'\|_0$.
\end{definition}

We also consider a partial DP equivalent of CDP.
\begin{definition}[Partial CDP and Per-Attribute Partial CDP]
Let $\varepsilon : \mathcal{X} \times \mathcal{X} \to \mathbb{R}$ be symmetric and non-negative.  We say that a randomized algorithm $M : \mathcal{X}^n \to \mathcal{Y}$ is \emph{$\varepsilon$-partially CDP ($\varepsilon$-$\nabla$CDP)} if, for all inputs $x,x'\in \mathcal{X}^n$ differing only on a single entry,
$\dr{*}{M(x)}{M(x')} \le \frac12 \varepsilon(x_i,x'_i)^2,$ where  $i \in [n]$ is the index on which $x$ and $x'$ differ (i.e., $\forall i' \ne i ~~ x_{i'} = x'_{i'}$).

For $\varepsilon_0 \in \mathbb{R}$, we define \emph{$\varepsilon_0$-per-attribute partial CDP ($\varepsilon_0$-$\nabla_0$CDP)} to be $\varepsilon$-$\nabla$CDP with $\varepsilon(x,x') := \varepsilon_0 \cdot \|x-x'\|_0$.
\label{defn:pcdp}
\end{definition}
Definition \ref{defn:pcdp} is a relaxation of Definition \ref{defn:ppdp}, i.e., $\varepsilon$-$\nabla$DP implies $\varepsilon$-$\nabla$CDP. CDP has better composition properties than pure DP, which makes it more useful in practice.\footnote{We do not consider approximate DP as a basis for partial DP, as approximate DP has poor group privacy properties, which are essential for our setting. Another option is Gaussian DP \cite{dong2019gaussian}, which has properties very similar to CDP.}

DP is usually defined in terms of neighboring datasets --- i.e., a binary relation, rather than a metric. We can define partial DP in terms of such a graph of neighboring records:
\begin{lemma}[Equivalent Definition of Partial DP]\label{lem:graph}
    Let $G$ be an undirected non-negatively weighted graph on $\mathcal{X}$. Let $M : \mathcal{X}^n \to \mathcal{Y}$ be a randomized algorithm. Suppose, for every pair $x,x'\in \mathcal{X}^n$ differing on a single entry $i \in [n]$, if $\{x_i,x'_i\}$ is an edge in the graph $G$ with weight $\varepsilon_0$, then $\dr{\infty}{M(x)}{M(x')} \le \varepsilon_0$ (resp., $\dr{*}{M(x)}{M(x')} \le \frac12\varepsilon_0^2$). Let $\varepsilon : \mathcal{X} \times \mathcal{X} \to \mathbb{R}$ be the distance metric on the graph $G$. Then $M$ is $\varepsilon$-$\nabla$DP (resp., $\varepsilon$-$\nabla$CDP).
\end{lemma}
In particular, Lemma \ref{lem:graph} tells us that per-attribute partial DP is equivalent to changing the neighboring relation of DP to consider changing only a single attribute of a person, rather than an entire person record. That is, an equivalent definition of $\varepsilon_0$-$\nabla_0$DP (or $\varepsilon_0$-$\nabla_0$CDP) is to require that for all pairs $x,x' \in (\mathcal{X}_1 \times \cdots \times \mathcal{X}_d)^n$ of datasets differing on a single attribute of a single person, we have $\forall S ~~\pr{}{M(x) \in S} \le e^{\varepsilon_0} \cdot \pr{}{M(x')}$ (or, respectively, $\dr{*}{M(x)}{M(x')} \le \frac12 \varepsilon_0^2$). I.e., the graph in Lemma \ref{lem:graph} is the Hamming graph with each edge having the same weight $\varepsilon_0$. We will define such pairs as \emph{neighboring}.

\subsection{Basic Properties of Partial DP}
An essential property of partial DP is that it is directly comparable to standard DP: %
\begin{proposition}\label{prop:ppdp-conversion}
    If $M : \mathcal{X}^n \to \mathcal{Y}$ satisfies $\varepsilon_0$-$\nabla_0$DP and $\mathcal{X} = \mathcal{X}_1 \times  \cdots \times \mathcal{X}_d$ consists of $d$ attributes, then $M$ satisfies $(d \cdot \varepsilon_0)$-DP. Similarly, $\varepsilon_0$-$\nabla_0$CDP implies $\frac12d^2\varepsilon_0^2$-zCDP, which in turn implies $(\tilde\varepsilon,\tilde\delta)$-DP for all $\tilde\varepsilon \ge \frac12 d^2 \varepsilon_0^2$ and $\tilde\delta = \exp\left(-(\tilde\varepsilon - \frac12 d^2 \varepsilon_0^2)^2 /2 d^2 \varepsilon_0^2 \right)$. %
    
    More generally, if $M : \mathcal{X}^n \to \mathcal{Y}$ satisfies $\varepsilon$-$\nabla$DP, then $M$ satisfies $(\sup_{x,x'\in\mathcal{X}} \varepsilon(x,x'))$-DP.
    Conversely, if $M$ satisfies $\varepsilon$-DP, then $M$ satisfies $\varepsilon$-$\nabla$DP (where we interpret $\varepsilon$ as a constant function) and $\varepsilon$-$\nabla_0$DP.
\end{proposition}

The conversion from per-attribute partial DP to standard DP is an application of the group privacy property (a.k.a. the triangle inequality for R\'enyi divergences). That is, $\dr{\infty}{P}{Q} \le \dr{\infty}{P}{R} + \dr{\infty}{R}{Q}$ and $\dr{*}{P}{Q} \le \left( \sqrt{\dr{*}{P}{R}} + \sqrt{\dr{*}{R}{Q}} \right)^2$ for all appropriate probability distributions $P$, $Q$, and $R$. 
More generally, if we have $\varepsilon_0$-$\nabla_0$DP and the adversary is interested in only $k$ attributes, then we obtain a privacy guarantee comparable to $(k\cdot\varepsilon_0)$-DP. 
This conversion may or may not be tight, but it is important that we can directly relate the partial DP guarantee back to standard DP.

Next we have composition, which is inherited from the standard DP definition.

\begin{lemma}[Sequential composition]
\label{lem:composition}
    Let $M_1 : \mathcal{X}^n \to \mathcal{Y}_1$ be $\varepsilon_1$-$\nabla$DP (respectively,  $\varepsilon_1$-$\nabla$CDP). Let $M_2 : \mathcal{X}^n \times \mathcal{Y}_1 \to \mathcal{Y}_2$ be such that the restriction $M_2(\cdot,y) : \mathcal{X}^n \to \mathcal{Y}_2$ is $\varepsilon_2$-$\nabla$DP (resp., $\varepsilon_2$-$\nabla$CDP) for all $y \in \mathcal{Y}_1$. Define $M_{12} : \mathcal{X}^n \to \mathcal{Y}_2$ by $M_{12}(x)=M_2(x,M_1(x))$. Then $M_{12}$ is $(\varepsilon_1+\varepsilon_2)$-$\nabla$DP (resp., $\sqrt{\varepsilon_1^2+\varepsilon_2^2}$-$\nabla$CDP).
\end{lemma}
A simple difference between per-attribute partial DP and standard DP is that if we run multiple DP algorithms on disjoint sets of attributes, then the privacy parameter does not grow with the number of attributes. In contrast, under standard DP, the privacy parameter would grow with the number of attributes following composition.
\begin{lemma}[Parallel composition, \cite{McSherry}]
    For $j \in [d]$, let $M_j : \mathcal{X}_j^n \to \mathcal{Y}_j$. %
    Let $\mathcal{X} = \mathcal{X}_1 \times \cdots \times \mathcal{X}_d$ and $\mathcal{Y} = \mathcal{Y}_1 \times \cdots \times \mathcal{Y}_d$. 
    Define $M : \mathcal{X}^n \to \mathcal{Y}$ by $M(x)_j = M_j((x_{i, j})_{i \in [n]})$ for all $x \in \mathcal{X}^n$ and $j \in [d]$, where $x_{i, j} \in \mathcal{X}_j$ denotes only the $j$th attribute of the $i$th record.
    If $M_j$ is $\varepsilon$-DP (resp., $\frac12\varepsilon^2$-zCDP) for each $j \in [d]$, then $M$ is $\varepsilon$-$\nabla_0$DP (resp., $\varepsilon$-$\nabla_0$CDP).
\end{lemma}
For example, if we release independent statistics about the medical records, browsing histories, employment, etc.~of a set of people and each individual release is $\varepsilon$-DP, then the overall release is $\varepsilon$-$\nabla_0$DP, assuming no overlapping attributes. In particular, this example characterizes the setting in which non-coordinating entities perform DP analysis on data from the same set of people.

In the rest of this paper, we provide several algorithms in the per-attribute partial DP framework.  We show that, in a variety of settings, there is a quantitative separation between partial DP and standard DP, i.e., we can provide $\varepsilon_0$-$\nabla_0$DP (or $\varepsilon_0$-$\nabla_0$CDP) with a small per-attribute privacy parameter $\varepsilon_0$, but it is not possible to provide $\varepsilon$-DP (or $\frac12\varepsilon^2$-zCDP) with a small per-person privacy parameter $\varepsilon$. We argue that, in such settings, it is more informative to give a small per-attribute guarantee $\varepsilon_0$, in addition to the large per-person guarantee $\varepsilon$.
\section{Related Work}\label{sec:related-work}

\para{Definitions.}
There is a vast literature on privacy definitions both before and since the introduction of DP~\cite{DworkMNS06};~\citet{DesfontainesP19}  catalog 225 DP variants that have been proposed. We only discuss the definitions most closely related to  partial DP. We organize these definitions into three categories: (i) more general than partial DP, (ii) the same as or similar to per-attribute partial DP, and (iii) special cases of partial DP, but different from per-attribute partial DP.

(i) \textsl{Notions more general than partial DP:}
\citet{chatzikokolakis2013broadening} define a notion of \emph{metric DP} or \emph{$d$-privacy}, where the indistinguishability guarantee is determined by a metric $d$ on the space of all input data\emph{sets} (as opposed to data \emph{points} in partial DP). An equivalent definition, dubbed \emph{Lipschitz privacy} was given by~\citet{koufogiannis2015optimality}.
Similarly, \emph{Pufferfish privacy}~\cite{KiferM14} and \emph{Blowfish privacy}~\cite{HeMD14} define a generalized notion of neighboring datasets (a.k.a.~``secrets''), which yields a generalization of DP by taking the metric to be proportional to the distance between datasets on the graph of neighboring datasets (cf.~Lemma \ref{lem:graph}). 
Unlike these prior works, we restrict the definition of partial DP to consider pairs of datasets that only differ on the record of a single person, rather than considering pairs of datasets in which the records of multiple people may change.  We consider this restriction to be an important feature of our definition, as it ensures that partial DP remains comparable to standard DP (Proposition~\ref{prop:ppdp-conversion}) and thus can still be interpreted as an individual privacy guarantee.
\emph{Metric DP} or \emph{context-aware} DP has also been studied in the context of local DP \cite{alvim2018local,acharya2020context}; our focus, however, is on central DP.

(ii) \textsl{Notions comparable to per-attribute partial DP:}
Our notion of per-attribute partial DP is equivalent to the definition of \emph{attribute DP} given by \citet{kifer2011no}; \citet{kenthapadi2012privacy} and \citet{ahmed2016social} also use this definition, but they simply call it ``differential privacy'' without qualification.
\citet{asi2019element} define \emph{element-level DP}, where the distance between data points is determined by the number of ``elements'' on which they differ; examples of ``elements'' are whether or not a certain word is included in a person's message history, or whether or not a domain is in the browsing history. 

\textsl{(iii) Other special cases of partial DP:}
\citet{andres2013geo} define \emph{geo-indistinguishability}, which is a special case of partial DP in which $\varepsilon(x,x')=\varepsilon_0 \cdot \|x-x'\|_2$, i.e., inputs are points in space and the privacy guarantee scales with the Euclidean distance. Another special case of partial DP is \emph{edge DP} in graphs~\cite{hay2009accurate}, where a person corresponds to a vertex that may have many incident edges, but privacy is only guaranteed on a per-edge basis. \emph{Label DP}~\cite{chaudhuri2011sample} is also a special case of partial DP, where $\varepsilon(x,x')=\infty$ if $x$ and $x'$ differ on any attribute other than the label in the training dataset of a supervised machine learning task.

It is common to assume that each person contributes one record to the dataset, but often a person may contribute multiple records. If we do not account for this, then we have a relaxed version of DP, which has been dubbed \emph{item-level DP} or \emph{record-level DP}~\citep{houssiau2022difficulty,MessingDHKMMNPSW20}. A recent line of work on ``user-level DP'' provides algorithms that ensure standard DP even when each user has multiple records~\citep{levy2021learning,ghazi2021user,liu2020learning,cummings2021mean}.

We remark that most of the related prior work considers definitions based on pure DP,\footnote{The only exception is element-level DP \citep{asi2019element}, which is based on  R\'enyi DP \citep{Mironov17}, a relaxation of CDP.} which is rarely used in practice due to its inferior composition properties.  Our work also considers concentrated DP; while this extension is straightforward, we believe it is important.%

\para{Algorithms.}
Although many different privacy definitions have been proposed, surprisingly few algorithms have been studied under these notions. To the best of our knowledge, all of the comparable prior algorithmic results are variants of adding Laplace or Gaussian noise scaled to a modified version of sensitivity that fits the definition.
Our main technical contribution is a deeper exploration of the algorithmic aspects of a more granular privacy analysis; we provide several algorithms for a variety of standard data analysis tasks and give a more granular privacy analysis for each.

\section{Answering Query Workloads}
\label{sec:gen_alg_results} 

In this section we consider the problem of releasing statistics that depend on overlapping sets of attributes.  In particular, we investigate releasing private answers to an arbitrary family $\mathcal{Q}$ of queries $\{ q_j: \mathcal{X} \rightarrow \mathbb{R} \}_{j=1}^m$. E.g., if $\mathcal{X}=\{0,1\}^d$, then $\mathcal{Q}$ could be all $k$-way parities or $k$-way conjunctions.
These examples of low-order marginals are some of the best-studied families of queries in the DP literature; in particular, they are known to be among the ``hardest'' families of queries \citep{BunUV14}.  We design partial DP mechanisms for answering such families of queries to contrast the partial DP bounds against the standard DP bounds. %

\subsection{Warmup: Average Error via Noise Addition}

Abusing notation, let $\mathcal{Q}(x) \in \mathbb{R}^m$ denote the vector $( \frac{1}{n} \sum_{i=1}^n q_j(x_i) )_{j=1}^m$ of answers. 
 Let $\Delta = \sup \{ \|\mathcal{Q}(x)-\mathcal{Q}(x')\|_2 : x,x'\in \mathcal{X}\}$ be the \emph{diameter} and let
$\Delta_0 = \sup \{ \|\mathcal{Q}(x)-\mathcal{Q}(x')\|_2 : x,x'\in \mathcal{X}, \|x-x'\|_0 \le 1\}$ be the \emph{partial diameter} of the set of possible answers.  

A simple algorithm for answering $\mathcal{Q}$ under standard DP with low mean squared error (MSE), is the \emph{Projection Mechanism}~\cite{NikolovTZ13,DworkNT15}, which adds Gaussian noise $\mathcal{N}(\cdot, \cdot)$ to the vector of all query answers and projects this noisy vector back to the set of answer vectors that are consistent with some valid input. 
This naturally extends to partial DP. 
\begin{theorem}\label{thm:projection_mechanism}
For $\sigma>0$ and $\mathcal{Q} : \mathcal{X} \to \mathbb{R}^m$, define the projection mechanism $M : \mathcal{X}^n \to \mathbb{R}^m$ as follows. On input $x \in \mathcal{X}^n$, compute $\mathcal{Q}(x) = \frac1n \sum_i^n \mathcal{Q}(x_i)$, sample ${Y} \gets \mathcal{N}\left( \mathcal{Q}(x) , \sigma^2 I\right)$, and output $\hat{Y} := \argmin_{y \in \mathsf{conv}(\{\mathcal{Q}(\check x) : \check x \in \mathcal{X}\})} \|y-{Y}\|_2$, where $\mathsf{conv}(\cdot)$ is the convex hull. 
Then $M$ simultaneously satisfies $\frac12\varepsilon^2$-zCDP and $\varepsilon_0$-$\nabla_0$CDP for $\varepsilon = \frac{\Delta}{\sigma n}$ and $\varepsilon_0 = \frac{\Delta_0}{\sigma n}$.
Furthermore, for all $x \in \mathcal{X}^n$,
\begin{equation}
    \ex{}{\!\frac{1}{m}\! \left\|M(x)\!-\!\mathcal{Q}(x)\right\|_2^2\!} \!\le\! \min\! \left\{  \sigma^2 ,\! \frac{\sigma \!\cdot\! \Delta \!\cdot\! \sqrt{2\log|\mathcal{X}|}}{m} \right\}.\label{eq:projection_mechanism_partial}
\end{equation}
\end{theorem}

Theorem \ref{thm:projection_mechanism} simply states that the partial DP guarantee of the projection mechanism scales with the partial diameter $\Delta_0$ in place of the diameter $\Delta$; the former could be much smaller, depending on $\cal Q$, as we show next.

Consider the cases of $k$-way (unsigned) conjunctions or parities on $\mathcal{X}=\{0,1\}^d$, which are families of $m={d \choose k}$ queries each. In both cases, $\Delta = \sqrt{m}$ and $\Delta_0 = \sqrt{{d-1 \choose k-1}} = \sqrt{m \cdot k / d}$.
From Theorem \ref{thm:projection_mechanism}, for a given noise scale $\sigma$ of the projection mechanism, the ratio of the per-person and per attribute privacy parameters is $\frac{\varepsilon}{\varepsilon_0} = \frac{\Delta}{\Delta_0} = \sqrt{\frac{d}{k}}$.
The error guarantee of the projection mechanism under standard DP is near-optimal \citep{BunUV14,BlasiokBNS19}. Thus we have a separation --- in this setting we can report a smaller per-attribute privacy parameter $\varepsilon_0$ than the attainable per-person privacy parameter $\varepsilon$.  

Alternatively, if we fix a $\varepsilon_0$-$\nabla_0$CDP guarantee, the MSE is \[\ex{}{\!\frac{1}{m}\! \left\|M(x)\!-\!\mathcal{Q}(x)\right\|_2^2\!} \!\le\! \min\left\{\frac{\Delta_0^2}{\varepsilon_0^2 n^2} , \frac{\Delta_0 \cdot \Delta \cdot \sqrt{2\log|\mathcal{X}|}}{\varepsilon_0 n m}\right\} = \min\left\{\frac{m k}{\varepsilon_0^2 n^2 d} , \frac{ \sqrt{k \cdot 2\log 2}}{\varepsilon_0 n}\right\}.\] In contrast, for $\frac12 \varepsilon^2$-zCDP the MSE is \[\ex{}{\!\frac{1}{m}\! \left\|M(x)\!-\!\mathcal{Q}(x)\right\|_2^2\!} \!\le\! \min\left\{\frac{\Delta^2}{\varepsilon^2 n^2} , \frac{\Delta^2 \cdot \sqrt{2\log|\mathcal{X}|}}{\varepsilon n m}\right\} = \min\left\{\frac{m}{\varepsilon^2 n^2} , \frac{ \sqrt{d \cdot 2\log 2}}{\varepsilon n}\right\}.\] 
That is to say that the MSE under per-attribute partial DP scales with $k$ (the number of attributes that each query can depend on) rather than $d$ (the total number of attributes). 

\begin{proof}[Proof of Theorem \ref{thm:projection_mechanism}]
    The privacy guarantee follows from the properties of Gaussian noise addition and the fact that the sensitivity of $\mathcal{Q}(x)$ with respect to changing one attribute is $\Delta_0/n$ in the 2-norm, along with the postprocessing property of DP. 
    Now we turn to the accuracy analysis, for which we fix an input $x \in \mathcal{X}^n$.
    We have $\ex{}{\|Y-\mathcal{Q}(x)\|_2^2} = m \cdot \sigma^2$ and we must convert this into a guarantee on $\ex{}{\|\hat Y-\mathcal{Q}(x)\|_2^2}$.
    Since $\hat{Y} := \argmin_{y \in \mathsf{conv}(\{\mathcal{Q}(\check x) : \check x \in \mathcal{X}\})} \|y-{Y}\|_2$ and $\mathcal{Q}(x) \in \mathsf{conv}(\{\mathcal{Q}(\check x) : \check x \in \mathcal{X}\})$ we have $\|\hat Y-\mathcal{Q}(x)\|_2 \le \| Y-\mathcal{Q}(x)\|_2$ with probability $1$. That is, the projection step can never increase the error. This yields the first term in the minimum \eqref{eq:projection_mechanism_partial}.
    Now we have a geometric claim: \[\langle \mathcal{Q}(x) - \hat Y , Y - \hat Y \rangle \le 0.\] This again holds with probability 1. (This is most easily seen by drawing a picture.  Formally, if this were false, then for some $0\le t<1$, setting $\hat{y} = (1-t)\mathcal{Q}(x)+t\hat{Y}$ yields $\|\hat{y}-{Y}\|_2 < \|\hat{Y}-{Y}\|_2$, a contradiction.)
    It follows that \[\|\hat Y - \mathcal{Q}(x)\|_2^2 \le \langle \hat Y - \mathcal{Q}(x) , Y - \mathcal{Q}(x) \rangle \le \sup_{\check y \in \mathsf{conv}(\{\mathcal{Q}(\check x) : \check x \in \mathcal{X}\})} \langle \check y - \mathcal{Q}(x) , Y - \mathcal{Q}(x) \rangle, \] and
    \begin{align*}
        \ex{}{\|\hat Y - \mathcal{Q}(x)\|_2^2} &\le \ex{}{\sup_{\check y \in \mathsf{conv}(\{\mathcal{Q}(\check x) : \check x \in \mathcal{X}\})} \langle \check y - \mathcal{Q}(x) , Y - \mathcal{Q}(x) \rangle }\\
        &= \ex{G \gets \mathcal{N}\left(0,\sigma^2 I\right)}{\sup_{\check y \in \{\mathcal{Q}(\check x) : \check x \in \mathcal{X}\}} \langle \check y - \mathcal{Q}(x) , G \rangle }\\
        &\le \sqrt{2 \log |\{\mathcal{Q}(\check x) : \check x \in \mathcal{X}\}|} \cdot  \sigma \cdot \sup_{\check y \in \{\mathcal{Q}(\check x) : \check x \in \mathcal{X}\}} \| \check y - \mathcal{Q}(x) \|_2\\
        &\le \sigma \cdot \Delta \cdot \sqrt{2\log |\mathcal{X}|}.
    \end{align*}
    The penultimate inequality follows from a union bound style argument for bounding the expectation of the maximum of a finite number of Gaussians \citep{kamath2015bounds}. 
\end{proof}

    The $\Delta \cdot \sqrt{\log |\mathcal{X}|}$ term in the guarantee of Theorem \ref{thm:projection_mechanism} can, in general, be replaced by the Gaussian width $\ex{G \gets \mathcal{N}(0,I)}{\sup_{y \in S_x} \langle y , G \rangle}$ of the set $S_x := \{\mathcal{Q}(\check x) - \mathcal{Q}(x) : \check x \in \mathcal{X}\}$.

\subsection{Maximum Error via an Iterative Algorithm}

Rather than average error, we can obtain bounds on the maximum error over all queries. In the standard DP setting, optimal bounds are given by \emph{Private Multiplicative Weights (PMW)} \citep{HardtR10} and its refinement the \emph{Multiplicative Weights Exponential Mechanism (MWEM)} \cite{HardtLM10}.
We now look at a per-attribute partial DP version of MWEM, $M_{\ell}: {\cal X}^n \rightarrow [0,1]^m$,  given in Algorithm \ref{alg:pdp-mwem}.

In this section we assume that the domain ${\cal X} = {\cal X}_1 \times \cdots\times {\cal X}_d$ of the queries is finite and the range of the queries is $[0, 1]$.  For $x \in {\cal X}^n$, let
$q(x) = \frac1n \sum_{i = 1}^n q(x_i)$ and for a distribution $D$ on ${\cal X}$, let
$q(D) = \ex{u \sim D}{q(u)}$.  

For each $q \in \mathcal{Q}$, let $\mathsf{attr}(q) \subseteq [d]$ denote the attributes that $q$ depends on, i.e., this satisfies the property that, for any $x,x' \in \mathcal{X}$, if $x_i = x'_i$ for all $i \in \mathsf{attr}(q)$, then $q(x)=q(x')$.

%
%

\begin{algorithm}[h]
\small
\caption{Partial DP MWEM, $M_{\ell}: {\cal X}^n \rightarrow [0,1]^m$.} \label{alg:pdp-mwem}
\begin{algorithmic}[1]
    \STATE \textbf{Input:} Private dataset $x \in \mathcal{X}^n$.
    \STATE \textbf{Parameters:} Set $\mathcal{Q}$ of $m$ queries defined by functions $q : \mathcal{X} \to [0,1]$,  privacy parameter $\varepsilon_0>0$, iterations $T \in \mathbb{N}$, and  queries per round $\ell \in \mathbb{N}$.
    \STATE Set $\varepsilon_T \gets \varepsilon_0/\sqrt{2T}$ and $A_1 \gets$  uniform distribution on $\mathcal{X}$.
    \FOR{ $t=1,\dots,T$:}
    \STATE Select $\{q_{t,1}, \dots, q_{t,\ell}\} \subseteq \mathcal{Q}$ with $\mathsf{attr}(q_{t,i}) \cap \mathsf{attr}(q_{t,j}) = \emptyset$ for all $1 \le i < j \le \ell$ in a $\varepsilon_T$-$\nabla_0$CDP manner using the exponential mechanism, 
i.e., \[\pr{}{(q_{t,1},\dots,q_{t,\ell})=(q_1,\dots,q_\ell)} \propto \exp\left({\varepsilon_T n} \sum_{i=1}^\ell |q_i(A_t)-q_i(x)| \right).\]
    \STATE Answer $q_{t,1}, \dots, q_{t,\ell}$ with $\varepsilon_T$-$\nabla_0$CDP by sampling $a_{t,i} \gets \mathcal{N}\left(q_{t,i}(x), \frac{1}{n^2\varepsilon_T^2}\right)$ for all $i \in [\ell]$. (But, if $a_{t,i} > 1$, set $a_{t,i} \gets 1$ and if $a_{t,i}<0$, set $a_{t,i} \gets 0$.)
    \STATE Define the distribution $A_{t+1}$ on $\mathcal{X}$ by \[A_{t+1}(u) \propto A_t(u) \cdot \exp\left(\frac12 \cdot \sum_{i=1}^{\ell} q_{t,i}(u) \cdot (a_{t,i} - q_{t,i}(A_t)) \right).\]
    \ENDFOR
    \STATE $A \gets \frac1T \sum_{t=1}^T A_t$.
    \STATE \textbf{Output:} $\ex{X \sim A}{\mathcal{Q}(X)}$.
\end{algorithmic}
\end{algorithm}

\begin{theorem}\label{thm:pdp-mwem}
    $M_\ell$ in Algorithm~\ref{alg:pdp-mwem} satisfies $\varepsilon_0$-$\nabla_0$CDP and simultaneously $\frac12 (\ell\cdot\varepsilon_0)^2$-zCDP and, for all inputs $x \in \mathcal{X}^n$, we have 
    \begin{equation}\label{eq:pdp-mwem-T}
    \ex{}{\max_{q_1,\dots, q_\ell \in \mathcal{Q} \atop \forall i<j ~\mathsf{attr}(q_i)\cap\mathsf{attr}(q_j)=\emptyset}  \frac1\ell \sum_{i=1}^\ell |q_i(A)-q_i(x)| } \le \sqrt{ \frac{2T}{n^2 \varepsilon_0^2} + \frac{4 \log |\mathcal{X}|}{T \cdot \ell} } + \frac{\sqrt{2T}}{\varepsilon_0 n} \log m.
    \end{equation} 
    If $T=\Theta\left(\frac{\sqrt{\log|\mathcal{X}|} \cdot \varepsilon n }{\sqrt{\ell} \cdot \log m}\right)$, then the RHS of (\ref{eq:pdp-mwem-T}) is
    $O\left(\sqrt{\frac{\sqrt{\log |\mathcal{X}|}\cdot \log m}{\sqrt{\ell} \cdot \varepsilon n}}\right)$.
\end{theorem}
Note that the LHS of (\ref{eq:pdp-mwem-T}) refers to the maximum average error over an arbitrary set of $\ell$ attribute-disjoint queries; while this is a weaker bound than the maximum error over all individual queries, it is stronger than MSE. We leave it as an open problem to close this gap and obtain a standard maximum error bound.

The main difference between Algorithm~\ref{alg:pdp-mwem} and standard MWEM is that in each iteration we sample $\ell$ attribute-disjoint queries in parallel, rather than a single query; setting $\ell=1$ recovers the standard MWEM algorithm and bound. For $\ell>1$, the RHS of (\ref{eq:pdp-mwem-T}) gets smaller, but the LHS also changes. 

For the case of $k$-way conjunctions on $\mathcal{X}=\{0,1\}^d$, setting $\ell=\lfloor d/k \rfloor$ obtains the error bound of $\alpha = O\left(k^{\frac34} \sqrt{\frac{\log d}{\varepsilon_0 n}}\right)$; in contrast, the corresponding bound under standard DP is 
$\alpha = O\left(k^{\frac12} d^{\frac14}  \sqrt{\frac{\log d}{\varepsilon_0 n}}\right)$.  Thus, the partial DP bound is 
near-independent of $d$ and instead depends mainly on $k$.

\begin{proof}[Proof of Theorem \ref{thm:pdp-mwem}]
    The privacy guarantee follows from composition over rounds (Lemma \ref{lem:composition}) and the privacy of Gaussian noise addition and the exponential mechanism \citep{DPorg-exponential-mechanism-bounded-range}. Note that $\sum_{i=1}^\ell |q_{t,i}(A_t)-q_{t,i}(x)|$ has sensitivity $1/n$ in the setting of partial DP as long as they are attribute disjoint, i.e., $\mathsf{attr}(q_{t,i}) \cap \mathsf{attr}(q_{t,j}) = \emptyset$ for all $1 \le i < j \le \ell$. That is, changing one attribute of one individual in $x$ can only change one of the $q_{t,i}(x)$ terms by $1/n$.
    The sensitivity with respect to changing an entire record is $\ell/n$, as each of the $\ell$ terms may change by $1/n$.
    
    Now we delve into the accuracy analysis. Fix an input $x \in \mathcal{X}^n$.  We use a potential function:
    \[\Psi(A) := \dr{1}{x}{A} = \sum_{u \in \mathcal{X}} x(u) \log\left(\frac{x(u)}{A(u)}\right),\]
    where we view $x$ as a probability distribution: $x(u) = \frac1n |\{i \in [n] : x_i = u\}|$.

    By the standard properties of KL divergence we have $\Psi(A_t) \ge 0$ for all $t \in [T+1]$ and we have $\Psi(A_1) \le \log |\mathcal{X}|$.
    Next we have a simple technical lemma \cite[Lemma A.4]{HardtLM10}:
    
    \begin{lemma}[Change in potential function]
    Let $A$ and $A'$ and $x$ be probability distributions on $\mathcal{X}$. Suppose \[\forall u \in \mathcal{X} ~~~~ A'(u) \propto A(u) \cdot \exp \left( \frac12 \cdot q(u) \cdot (a-q(u)) \right),\] where $q : \mathcal{X} \to [0,1]$ and $a \in [0,1]$.
    Define $\Psi(A)=\sum_{u \in \mathcal{X}} x(u) \log(x(u)/A(u))$ and, similarly, $\Psi(A')=\sum_{u \in \mathcal{X}} x(u) \log(x(u)/A'(u))$.
    Then \[\Psi(A)-\Psi(A') \ge \frac14 (q(A)-q(x))^2 - \frac14 (a-q(x))^2.\]
    \end{lemma}
    
    We also need the following claim: 
    Suppose  $\{q_{t,1}, \cdots, q_{t,\ell}\} \subseteq \mathcal{Q}$ satisfy $\mathsf{attr}(q_{t,i}) \cap \mathsf{attr}(q_{t,j}) = \emptyset$ for all $1 \le i < j \le \ell$. Let $A_t=A_{t,0}$ be an arbitrary probability distribution on $\mathcal{X}$.
    For $i \in [\ell]$, iteratively define a probability distribution $A_{t,i}$ on $\mathcal{X}$ by $A_{t,i}(u) \propto A_{t,i-1}(u) \cdot \exp\left(\frac12 \cdot q_{t,i}(u) \cdot (a_{t,i} - q_{t,i}(A_t)) \right)$.
    Then $A_{t,\ell}(u) = A_{t+1}(u) \propto A_t(u) \cdot \exp\left(\frac12 \cdot \sum_i^{\ell} q_{t,i}(u) \cdot (a_{t,i} - q_{t,i}(A_t)) \right)$ for all $u$.
    Intuitively the claim says that, because the queries we select in each round are attribute-disjoint, we can treat the update using the sum of the queries as identical to a series of sequential updates.
    
    Putting everything together we have \[\log |\mathcal{X}| \ge \Psi(A_1)-\Psi(A_{T+1}) = \sum_{t=1}^T \Psi(A_t)-\Psi(A_{t+1}) \ge \frac14 \sum_{t = 1}^T \sum_{i = 1 }^\ell (q_{t,i}(A_t)-q_{t,i}(x))^2 - (a_{t,i}-q_{t,i}(x))^2.\]
    For all $i \in [\ell]$ and $t \in [T]$, the quantity $a_{t,i}-q_{t,i}(x)$ is distributed according to a truncated version of $\mathcal{N}(0,1/n^2\varepsilon_T^2)$. Thus $\ex{}{(a_{t,i}-q_{t,i}(x))^2} \le 1/n^2\varepsilon_T^2 = 2T/n^2\varepsilon_0^2$. Rearranging and applying Jensen's inequality gives the bound \[\ex{}{\frac1T \sum_{t = 1}^T \frac1\ell \sum_{i = 1 }^\ell |q_{t,i}(A_t)-q_{t,i}(x)|}^2 \le \ex{}{\frac1T \sum_{t = 1}^T \frac1\ell \sum_{i = 1 }^\ell (q_{t,i}(A_t)-q_{t,i}(x))^2} \le \frac{2T}{n^2 \varepsilon_0^2}+ \frac{4 \log |\mathcal{X}|}{T \cdot \ell}.\]
    
    Next we invoke the accuracy guarantee of the exponential mechanism \cite[Lemma 7.1]{BassilyNSSSU16}:
    \[\forall t ~~~ \ex{}{\sum_{i=1}^\ell |q_{t,i}(A_t)-q_{t,i}(x)|} \ge \max_{q_1,\cdots, q_\ell \in \mathcal{Q} \atop \forall i<j ~\mathsf{attr}(q_i)\cap\mathsf{attr}(q_j)=\emptyset} \sum_{i=1}^\ell |q_i(A_t)-q_i(x)| -  \frac{1}{\varepsilon_T n} \log |\mathcal{Q}^\ell|.\]
    
    Now we have
    \begin{align*}
        \ex{}{\max_{q_1,\cdots, q_\ell \in \mathcal{Q} \atop \forall i<j ~\mathsf{attr}(q_i)\cap\mathsf{attr}(q_j)=\emptyset}  \frac1\ell \sum_{i=1}^\ell |q_i(A)-q_i(x)| } &= \ex{}{ \max_{q_1,\cdots, q_\ell \in \mathcal{Q} \atop \forall i<j ~\mathsf{attr}(q_i)\cap\mathsf{attr}(q_j)=\emptyset}  \frac1\ell \sum_{i=1}^\ell \left| \frac1T \sum_{t=1}^T q_i(A_t)-q_i(x) \right| } \\
         &\le \ex{}{ \frac1T \sum_{t=1}^T \max_{q_1,\cdots, q_\ell \in \mathcal{Q} \atop \forall i<j ~\mathsf{attr}(q_i)\cap\mathsf{attr}(q_j)=\emptyset}  \frac1\ell \sum_{i=1}^\ell \left| q_i(A_t)-q_i(x) \right| }\\
         &\le \ex{}{ \frac1T \sum_{t=1}^T \frac1\ell \sum_{i=1}^\ell |q_{t,i}(A_t)-q_{t,i}(x)| + \frac{1}{\varepsilon_T n \ell} \log |\mathcal{Q}^\ell| }\\
         &\le \sqrt{ \frac{2T}{n^2 \varepsilon_0^2} + \frac{4 \log |\mathcal{X}|}{T \cdot \ell} } + \frac{\sqrt{2T}}{\varepsilon_0 n} \log |\mathcal{Q}|.
         \qedhere
    \end{align*}
\end{proof}

We remark that both the projection mechanism and the multiplicative weights exponential mechanism are not polynomial time algorithms in general. This limitation is not specific to partial DP. Under standard cryptographic assumptions it is known to be computationally infeasible to generate synthetic data (as Algorithm \ref{alg:pdp-mwem} does) even for $2$-way conjunctions \cite{ullman2011pcps}. This computational hardness is also not specific to synthetic data \cite{kowalczyk2018hardness}. However, for the special case of $k$-way conjunctions or parities it remains an open problem to devise a polynomial-time algorithm that comes close to matching the guarantees of MWEM, or to prove an impossibility result. There has been some limited progress on devising efficient versions of the projection mechanism \cite{dwork2014using}.
\section{Histograms, Heavy Hitters \& Applications}\label{sec:heavy_hitters}

We consider the fundamental problem of computing a histogram or, equivalently, of computing the heavy hitters, which is well-studied in DP in various models and settings \cite{DworkMNS06,MishraS06,KorolovaKMN09,HardtT10,HsuKR12,ErlingssonPK14,BassilyS15,BunNS19,BunNS19hhldp,BalcerC20,BalcerCJM21}.

\begin{definition}[Histogram Problem]
In the \emph{histogram} problem, the input dataset consists of $x_1, \dots, x_n \in \{0, 1\}^d$; the frequency of an element $y \in \{0, 1\}^n$ is defined as $f_y := |\{i \in [n] : x_i = y\}|$.  An algorithm is said to solve the histogram problem with (normalized $\ell_\infty$) error $\nu \in (0, 1)$ if it outputs $\{\hf_y\}_{y \in \{0, 1\}^d}$ such that $\max_y |\hf_y - f_y| \leq \nu n$.
\end{definition}

Computing histograms is generally an easy task as far as privacy is concerned. (Although it can be challenging when we combine privacy with computational or communication constraints in a distributed setting.)
The canonical algorithm is to add independent noise to each count $f_y$ and apply parallel composition -- each individual only contributes to one count.
To attain pure DP we would add Laplace noise. For Concentrated DP, we would add Gaussian noise. And, for approximate DP, we could add truncated Laplace or Gaussian noise.

However, this canonical algorithm very closely resembles the worst-case algorithm envisaged by the DP definition.
Suppose an adversary wishes to determine whether or not $x_i=y$. That is, the adversary seeks to learn one bit about individual $i$.
Often the histogram is sparse, so no other individual $j \in [n] \setminus\{i\}$ has $x_j=y$. Thus it suffices for the adversary to figure out whether $f_y=1$ or whether $f_y=0$ .
If we add Laplace noise to attain $\varepsilon$-DP with large $\varepsilon$, then the adversary can easily distinguish between these two cases. Namely, we would add $\xi_y \gets \mathsf{Laplace}(1/\varepsilon)$ to $f_y$ and $\pr{}{|\xi_y|\ge\frac12} = e^{-\varepsilon/2}$. So, with probability $1-e^{-\varepsilon/2}$, rounding the private value $f_y+\xi_y$ to the nearest integer returns the non-private value $f_y$. In particular, if $\varepsilon=10$, rounding returns the true count with probability $>99\%$.

Given that histograms are a well-studied problem and the canonical algorithm yields the kind of worst-case privacy outcomes that we want to avoid, it is natural to ask whether we can design a per-attribute partial DP algorithm for histograms that avoids this worst-case behaviour. In terms of our intuitive justification for tolerating large $\varepsilon$, the canonical algorithm is not ``nice'' and the question is whether we can devise a ``nice'' algorithm for histograms.

We restrict our attention to pure DP, both for simplicity and because this highlights the strength of our results -- we are able to obtain these results under the most stringent form of DP. Naturally, our methods can be extended to Concentrated DP etc.

We provide the following result. We also provide a nearly-matching lower bound later in Theorem \ref{thm:heavy-hitters-lb}.

\begin{theorem} \label{thm:heavy-hitters-main}
Let $n \geq O\left(\frac{1}{\eps\nu} \cdot \log(d/\eta) \cdot \log(1/\nu)\right)$. Then, there exists an $\eps$-\PDP algorithm that with probability $1 - \eta$ solves the histogram problem with error $\nu$. Moreover, the algorithm runs in expected time $\poly(nd/(\eta \nu))$.
\end{theorem}

This sample complexity bound should be compared with the one in the standard $\eps$-DP setting, which is $n = \Theta\left(\frac{d}{\eps \nu}\right)$.\footnote{The standard $\frac12\varepsilon^2$-zCDP sample complexity is $n =\Theta\left(\frac{\sqrt{d}}{\eps \nu}\right)$ and under $(\varepsilon,\delta)$-DP it is $n = \Theta\left(\frac{\min\{d,\log(1/\delta)\}}{\eps \nu}\right)$.} In other words, there is an exponential separation (in terms of the dimension) between the standard DP parameter and the per-attribute partial DP parameter. 

\subsection{Our Algorithm}

In this section we design a novel algorithm for the histogram problem under per-attribute partial DP.  
Our algorithm will in fact output a succinct representation of $\{\hf_y\}_{y \in \{0, 1\}^d}$ by outputting a list $L \subset \{0, 1\}^d$ and $\{\hf_y\}_{y \in L}$, where $\hf_y = 0$ for every $y \notin L$ (this is needed for efficiency).

The main component of our algorithm is an $\eps$-\PDP algorithm for finding \emph{heavy hitters}---i.e., a list $L$ that is ``not too large'' and contains all $y$ with $f_y \geq \nu n$.
The algorithm builds \emph{a binary tree over attributes}, i.e., each node of the tree corresponds to an interval $I \subseteq [d]$ whose length is a power of two. Each node stores a list $L_I$ of heavy hitters among $x_1,\dots,x_n$ restricted to the attributes indexed by $I$. Each leaf stores heavy hitters of a single attribute; then the next level nodes store heavy hitters among pairs of attributes; and the root has the final list of heavy hitters. The list $L_I$ can be constructed by estimating the substring frequencies of $y_1 \circ y_2$ for a shortlist of all pairs $y_1, y_2$ that appears as heavy hitters of its children. The key is that each attribute $j \in [d]$ only appears in $\log d$ intervals---one in each level of the tree. This means that we only have to divide the privacy budget over $\log d$ levels, resulting in a noise of roughly $O(\log(d) / \eps)$ per level. (Under standard DP, we would need to apply composition over all $2d-1$ nodes, yielding noise scale $O(d/\eps)$).

Unfortunately, implementing the binary tree directly requires sample complexity $n \geq O_{\eta, \nu}((\log d)^2 / \eps)$ because we need to take a union bound over all $2d-1$ nodes to bound the probability that a heavy hitter is erroneously dropped, which contributes another factor of $\log d$ in addition to the one we get from composition over levels. To get from here to $O_{\eta, \nu}(\log(d)/ \eps)$ as claimed in \Cref{thm:heavy-hitters-main}, we employ yet another technique to save on privacy loss introduced by\citet{ZhangXX16} for their \emph{PrivTree} algorithm. The idea is roughly that when a count is far above the threshold, the privacy loss is actually much smaller than usual. Therefore, they introduce ``biasing'' and capping techniques, which can be thought of as lowering the threshold by a certain amount $\mu$ at each level and clipping the count to be at least the threshold minus $\mu$, respectively. %
We employ these ideas to shave a $\log d$ factor.

To formally describe our algorithm, we assume, without loss of generality, that $d$ is a power of two and define several additional notations:
\begin{itemize}[nosep]
\item For any $I := \{a, \dots, b\}$, we let $\Ileft := \left\{a, \dots, \lfloor \frac{a+b}{2} \rfloor\right\}$ and $\Iright := \left\{\lfloor \frac{a+b}{2} \rfloor + 1, \dots, b\right\}$.
\item For every $I := \{a, \dots, b\}$ and $s \in \{0, 1\}^{|I|}$, we define the frequency of $s$ w.r.t. position $I$ as $f^I_s := |\{i \in [n] \mid x^i|_I = s\}|$.
\item For every $\ell \in [\log d]$, let $\cI_\ell$ denote the collection of all sets $\{2^\ell (t - 1) + 1, \dots, 2^\ell t\}$ where $t \in [d / 2^\ell]$. Furthermore, we let $\cI := \bigcup_{\ell \in [\log d]} \cI_\ell$.
\item For two sets $S_1, S_2$ of strings, let $S_1 \circ S_2$ denote the set of strings resulting from concatenating an element of $S_1$ and an element of $S_2$, i.e., $S_1 \circ S_2 := \{s_1s_2 \mid s_1 \in S_1, s_2 \in S_2\}$.
\end{itemize}
\Cref{alg:priv-tree} contains the complete description. Here, $\Lap(\cdot)$ is the Laplace noise.   
It is worth noting that we will eventually choose $\lambda = O(1/\eps)$ and $\mu = O(\lambda \cdot \log(1/\nu))$ where each big-O notation hides a sufficiently large constant.  Our algorithm for histogram has the following guarantee:

\begin{wrapfigure}{R}{0.55\textwidth}
\centering
\vspace*{-4mm}
\begin{minipage}{0.44\textwidth}
\begin{algorithm}[H]
\small
\caption{PrivTree-Based Heavy Hitters.} \label{alg:priv-tree}
$\textsc{PrivHeavyHitter}$
\begin{algorithmic}[1]
\STATE{\bf Inputs:} $x^1, \dots, x^n$. 
\STATE{\bf Parameters:} $\lambda, \tau, \mu > 0$.
\FOR{$j \in [\log d]$}
\STATE $L_{\{j\}} \leftarrow \{0, 1\}$.
\ENDFOR
\FOR{$\ell \in [\log d]$}
\STATE $\tau_\ell \leftarrow \tau + (\ell - 1)\mu$.
\hfill \COMMENT{Threshold}
\FOR{$I \in \cI_\ell$}
\STATE $L_I \leftarrow \emptyset$.
\FOR{$s \in L_{\Ileft} \circ L_{\Iright}$}
\STATE $\hf^I_s \leftarrow \max\{f^I_s, \tau_\ell - \mu\} + \Lap(\lambda)$.
\IF{$\hf^I_s > \tau_\ell$}
\STATE Add $s$ to $L_I$.
\ENDIF
\ENDFOR
\ENDFOR
\ENDFOR
\STATE{\textbf{Output:}} $L_{[d]}$.
\end{algorithmic}
\end{algorithm}
\end{minipage}
\end{wrapfigure}

\subsection{Analysis}

We provide an overview of the privacy and utility analysis of Algorithm~\ref{alg:priv-tree}. 
The complete analysis is in the supplementary material.

\para{Privacy.}
For clarity, below we write the frequencies and lists as functions of the input datasets $\bD$ or $\bD'$.  We first show the following: suppose $\lambda, \mu$ are such that $\mu > 1$ and $\frac{2}{\lambda}\left(1 + \frac{1}{1 - e^{-\mu/\lambda}} \right) \leq \eps$, then \textsc{PrivHeavyHitter} is $\eps$-\PDP.  In fact, we will prove that even outputting all the sets $(L_I)_{I \in \cI}$ is $\eps$-\PDP, i.e., we show 
that for any neighboring datasets $\bD, \bD'$ and any values of $(S_I)_{I \in \cI}$, it holds that $\Pr\left[\forall I \in \cI, L_I(\bD) = S_I\right] \leq e^{\eps} \cdot \Pr\left[\forall I \in \cI, L_I(\bD') = S_I\right].$

To prove this statement, it suffices to consider the case where the differing elements in $\bD$ and $\bD'$ are $0_d$ and $10_{d - 1}$ respectively.  We let $I^\ell := \{1, \dots, 2^{\ell}\}$ and write $\hf_s$ and $f_s$ to mean $\hf_s^{I^\ell}$ and $f_s^{I^\ell}$, for notational ease.  We show that $\frac{\Pr\left[\forall I \in \cI, L_I(\bD) = S_I\right]}{\Pr\left[\forall I \in \cI, L_I(\bD') = S_I\right]}
$ is equal to
\begin{align*}
\prod_{\ell \in [\log d]} \prod_{s \in S_{I^\ell} \atop s \in \{0_{2^\ell}, 10_{2^{\ell - 1}}\}} \frac{\Pr[\hf_s(\bD) > \tau_\ell]}{\Pr[\hf_s(\bD') > \tau_\ell]} \times
\prod_{\ell \in [\log d]} \prod_{s \in (S_{\Ileft^\ell} \circ S_{\Iright^\ell}) \setminus S_{I^\ell} \atop s \in \{0_{2^\ell}, 10_{2^{\ell - 1}}\}} \frac{\Pr[\hf_s(\bD) \leq \tau_\ell]}{\Pr[\hf_s(\bD') \leq \tau_\ell]},
\end{align*}
and bound each of the RHS terms by $e^{\varepsilon/2}$.  To bound the first term, since $f_{01_{2^\ell - 1}}(\bD) < f_{01_{2^\ell - 1}}(\bD')$, we have $\Pr[\hf_{01_{2^\ell - 1}}(\bD) > \tau_\ell] \leq \Pr[\hf_{01_{2^\ell - 1}}(\bD') > \tau_\ell]$. Therefore, it suffices to bound $\prod_{\ell \in [\log d] \atop 0_{2^\ell} \in S^{{I^\ell}}_\ell} \frac{\Pr[\hf_{0_{2^\ell}}(\bD) > \tau_\ell]}{\Pr[\hf_{0_{2^\ell}}(\bD') > \tau_\ell]}$.
Let $L_0$ be the smallest integer such that $f_{0_{2^{L_0}}}(\bD) > \tau_{L_0} - \mu$.  (i) For all $\ell < L_0$ we have $\min\{\hf_{0_{2^\ell}}(\bD), \tau_\ell - \mu\} = \tau_\ell - \mu = \min\{\hf_{0_{2^\ell}}(\bD'), \tau_\ell - \mu\}$, and hence $\forall \ell < L_0, \frac{\Pr[\hf_{0_{2^\ell}}(\bD) > \tau_\ell]}{\Pr[\hf_{0_{2^\ell}}(\bD') > \tau_\ell]} = 1$.
(ii) For $\ell = L_0$, notice that $\min\{\hf_{0_{2^\ell}}(\bD), \tau_\ell - \mu\} - \min\{\hf_{0_{2^\ell}}(\bD'), \tau_\ell - \mu\} \leq 1$.  Hence, by the DP property of the Laplace mechanism, 
$
\frac{\Pr[\hf_{0_{2^{L_0}}}(\bD) > \tau_{L_0}]}{\Pr[\hf_{0_{2^{L_0}}}(\bD') > \tau_{L_0}]} \leq e^{1/\lambda}$.
(iii) For $\ell > L_0$, $f_{0_{2^{\ell}}}(\bD') > \tau_\ell + (\ell - L_0)\mu - 1$. Together,
\begin{align*}
\frac{\Pr[\hf_{0_{2^{\ell}}}(\bD) > \tau_{\ell}]}{\Pr[\hf_{0_{2^{\ell}}}(\bD') > \tau_{\ell}]}
&\leq \frac{\Pr[(\ell - L_0)\mu + \Lap(\lambda) > 0]}{\Pr[(\ell - L_0)\mu - 1 + \Lap(\lambda) > 0]}
&\leq \exp\left(\frac{1}{\lambda} \cdot \exp\left(\frac{1 - (\ell - L_0)\mu}{\lambda}\right)\right),
\end{align*}
where the last step uses~\cite[Lemma 2.1]{ZhangXX16}.  From these inequalities, the bound of
$e^{\varepsilon/2}$ on the first term follows from our choice of $\lambda$, $\mu$.  The bound for the second term follows similarly.

\para{Utility.}
We need to show that we discover all heavy hitters and that the expected list size is small; the latter will also imply that the expected running time of the algorithm is small.  Let $\eta, \nu \in (0, 0.1]$. Suppose that $\tau = 0.5\nu n$, $\tau \geq 8 \mu \log d + 8 \lambda \log(d/(\eta \nu))$ and $\mu \geq \lambda \log(16/\nu)$. We show

(i) Heavy hitters discovered: W.p. $1 - 0.5\eta$, $L_{[d]}$ contains all $s \in \{0, 1\}^d$ such that $f_s \geq 2\tau$.

(ii) Expected list size: $\E[|L_{[d]}|] \leq 8 / \nu$.

(iii) The expected running time of the algorithm is $\poly(d/\nu)$; this will follow from (ii).

To show (i), we fix any $s$ such that $f_s \geq 2\tau$. We will prove that $\Pr[s \notin L] \leq 0.5 \eta \nu$; since there are at most $1 / \nu$ such $s$'s, a union bound yields the claimed result.  Indeed,
\begin{align*}
&\Pr[s \notin L_{[d]}] = \Pr\left[\exists I \in \cI, s|_I \notin L_I\right]
\leq \sum_{I \in \cI} \Pr[s|_I \notin L_I \mid s|_{\Ileft} \in L_{\Ileft} \wedge s|_{\Iright} \in L_{\Iright}] \\
&\leq \sum_{I \in \cI} \Pr[2\tau + \Lap(\lambda) \leq \tau + \mu \log d]
\leq \sum_{I \in \cI} \Pr[0.5\tau + \Lap(\lambda) \leq 0]\\ 
&\leq \sum_{I \in \cI} \exp(-0.5\tau / \lambda) / 2 
\leq 2d \cdot \exp(-0.5\tau / \lambda)/2 \enspace \leq \enspace 0.5 \eta \nu.
\end{align*}
We show (ii), by an induction on $\ell$ that $\E[|L_I|] \leq 8 / \nu$ for all $\ell \in \{0, \dots, \log d\}$ and $I \in \cI_\ell$.  Consider any $\ell \in [\log d]$ and assume that the inductive hypothesis holds for $\ell - 1$ and any $I' \in \cI_{\ell - 1}$. Consider any $I \in \cI_\ell$ and let $S := \{s \in \{0, 1\}^{2^\ell} \mid f^I_s > 0.5\tau\}$. Note that $|S| < n/(0.5\tau) = 4/\nu$ and for any $s \notin S$, we have $f^I_s \leq \tau_\ell - \mu$.  Hence, for any $s \in (S_{\Ileft} \circ S_{\Iright}) \setminus S$, we have $\Pr[s \in L_I] \leq \nu / 16.$ $(*)$ 
\begin{align*}
\E[|L_I|] 
&= \E[|S \cap L_I|] + \E[|L_I \setminus S|] \leq |S| + \sum_{s \in \{0, 1\}^{2^\ell} \setminus S} \Pr[s \in L_I] \\
&\overset{(*)}{\leq} 4 / \nu + \sum_{s \in \{0, 1\}^{2^\ell} \setminus S} \nu / 16 \cdot \Pr[s \in S_{\Ileft} \circ S_{\Iright}] \\
&\leq 4 / \nu + \nu / 16 \cdot \E[|L_{\Ileft} \circ L_{\Iright}|]
\leq 4 / \nu + \nu / 16 \cdot (8 / \nu)(8 / \nu) \enspace = \enspace 8 / \nu,
\end{align*}
where the last inequality follows from induction.

\subsection{Lower Bound}

In this section we show that the $\log d$ dependence in Theorem~\ref{thm:heavy-hitters-main} necessary; our lower bound matches the upper bound up to the dependence on $\eta$ and $\log(1/\nu)$:

\begin{theorem} \label{thm:heavy-hitters-lb}
Assume that $d \geq 10 e^{1.1\eps}$. If there exists an $\eps$-\PDP algorithm that with probability $0.1$ solves the histogram problem with error $\nu$, then $n \geq \Omega\left(\frac{1}{\eps\nu}\log d\right)$.
\end{theorem}
The constant $0.1$ in Theorem~\ref{thm:heavy-hitters-lb} was chosen for concreteness, but a similar statement holds for any positive constant.  To prove Theorem \ref{thm:heavy-hitters-lb}, we follow the same packing-based approach that was used for the standard DP lower bound~\cite{HardtT10}. The difference is that our packing consists of only $d$ one-hot vectors (instead of all of $\{0, 1\}^d$); since the one-hot vectors are at Hamming distance only $2$ apart, the rest of the proof proceeds as before. 

\subsection{Applications of Histograms}

Algorithms for histograms are often used as subroutines for other algorithms.  As a concrete application of~\Cref{thm:heavy-hitters-main}, we obtain partial DP algorithms for the problems of PAC learning point functions and threshold functions, and discrete distribution estimation with $\ell_2^2$ error.%

\begin{theorem}\label{thm:point-and-threshold-func-intro}
For every $\eps, \alpha > 0$, $d\in\mathbb{N}$, there exists an $\eps$-$\nabla_0$DP proper PAC learner with error at most $\alpha$ and with sample complexity $n=\tilde{O}\left(\frac{1}{\alpha\eps} \log d\right)$ for point functions and threshold functions.
\end{theorem}

Theorem~\ref{thm:point-and-threshold-func-intro} should be contrasted with the sample complexity of proper PAC learning of point functions and threshold functions in the standard $\varepsilon$-DP setting, both of which are $n=\Theta\left(\frac{d}{\eps\alpha}\right)$~\cite{BeimelNS19,FeldmanX15}.

\begin{theorem} \label{thm:dist-est-l2-intro}
For every $\eps > 0$ and $n,d\in\mathbb{N}$, there exists an $\eps$-$\nabla_0$DP algorithm for discrete distribution learning whose $\ell_2^2$ error is $\tilde{O}\left(\frac{\log d}{\eps n} + \frac{1}{n}\right)$ with probability at least $0.9$.
\end{theorem}

In contrast, for standard DP, a packing lower bound~\cite{HardtT10} shows that even getting an $\ell_2^2$ error of $0.1$ (with constant probability) requires $n \geq \Omega(d/\eps)$.%

\section{Robust Learning of Halfspaces}
\label{sec:robust-halfspaces}

We next consider the problem of robust learning of halfspaces, under the (normalized) Hamming distance. 
A \emph{halfspace} is a function $h_{\bw} : \mathbb{R}^d \to \{-1,+1\}$ where $\bw \in \R^d$ and is defined as $h_{\bw}(x) = \sgn(\left<\bw, x\right>)$. We consider the class $\cH_{\hs} := \{h_{\bw} \mid \bw \in \R^d\}$ of all halfspaces.
Our input dataset consists of pairs $(x_1,y_1), \dots, (x_n,y_n) \in \mathcal{X} = \{-1,+1\}^{d} \times \{-1,+1\}$ drawn i.i.d.~from a distribution $\cD$ and our goal is to output a halfspace $h_{\bw} \in \cH_{\hs}$ that mislabels as small fraction of points w.r.t.~$\cD$ as possible---i.e., minimizing $\pr{(x, y) \leftarrow \cD}{h_{\bw}(x)\ne y}$.
We are interested in \emph{robust} learning. A sample $(x, y)$ is \emph{$\gamma$-robustly classified} if the hypothesis assigns all points in a $\gamma d$-radius Hamming ball around $x$ to the label $y$---i.e., $\forall \check x \in \{0,1\}^d ~ \|\check x - x\|_0 \le \gamma d \implies h_{\bw}(\check x) = y$. More formally, the $\gamma$-robust error is defined as follows: 
\begin{definition}[Robust Error]
For a distribution $\cD$ on $\mathcal{X} = \{-1,+1\}^d \times \{-1,+1\}$ and a hypothesis $h :\{-1,+1\}^d \to \{-1,+1\}$, we define its
\emph{$\gamma$-(Hamming-)robust error} to be \[\cR_\gamma(h, \cD) := \pr{(x, y) \gets \cD}{\exists \check x ~ h(\check x) \ne y \wedge \|\check x - x\|_0 \le d\gamma}.\] %
\end{definition}

The goal is to find a halfspace whose $\gamma'$-robust error is not much more than the optimal $\gamma$-robust error. Here $\gamma' < \gamma$ represents a relaxation in the margin that we pay for privacy, along with a relaxation in accuracy. %
The problem is well understood both in the non-private setting~\cite{DiakonikolasKM20} and in the standard DP setting~\cite{GhaziRMN21}. We give a partial DP algorithm:

\begin{theorem} \label{lem:robust-halfspaces}
Let $\eps, \gamma, \gamma' \in (0, 1]$ such that $\gamma > \gamma'$. There is an $\eps$-\PDP algorithm $M : \mathcal{X}^n \to \cH_{\hs}$ such that the following holds.
Let $\cD$ be a distribution on $\mathcal{X}$ and let $S \gets \cD^n$. Then
\begin{align*}
    & \ex{}{\cR_{\gamma'}(M(S), \cD)} \enspace \le \inf_{h \in \cH_{\hs}} \cR_{\gamma}(h,\cD) & \quad + \quad O\left(\frac{1}{\eps \sqrt{n} (\gamma - \gamma')^2} + \frac{\log(1/(\gamma - \gamma'))}{\eps^2 n \cdot (\gamma - \gamma')} \right).
\end{align*}
\end{theorem}%
Note that our error bound is independent of the number of attributes $d$, while in the standard DP setting, the error must grow linearly with $d$~\cite{GhaziRMN21}.

Our algorithm first privatizes the label $y \in \{-1,+1\}$ via Randomized Response; then, we can focus on a learner that is private in terms of the features $x \in \{-1,+1\}^d$. Our learner is in fact an instantiation of the exponential mechanism~\cite{McSherryT07} except we do \emph{not} apply it directly with respect to the (empirical) robust error, because for $x$ that is exactly at distance $\gamma d$ from the decision boundary, changing a single coordinate of $x$ could make it be considered mislabeled under $\gamma$-robust error. Instead, we smoothen the loss based on how far $x$ is from the decision boundary, similarly to the popular hinge loss, in order to reduce the sensitivity which gives the desired result.
To describe our algorithm, it will be most clear to separate the privacy of the labels and the privacy of the samples. In this regards, we say that an algorithm is \emph{$\eps$-sample-\PDP} if the DP guarantee is only enforced on changing a single coordinate of a sample. (There is no privacy guarantee on the labels.) 

For $\gamma > \gamma' > 0$ and $\alpha > 0$, we also say that a mechanism $M$ is \emph{$(\gamma', \gamma)$-learner with excess loss $\alpha$} iff $\ex{}{\cR_{\gamma'}(M(S), \cD)} \enspace \le \inf_{h \in \cH_{\hs}} \cR_{\gamma}(h,\cD) + \alpha$.

\subsection{From Sample-Only Privacy to Sample-and-Label Privacy}

A first observation is that by using randomized response on the labels, we can immediately translate an $\eps$-sample-\PDP algorithm to that of $\eps$-\PDP.

\begin{lemma} \label{lem:label-rr}
Let $\cH$ be any hypothesis class and $\eps \in (0, 1]$. Suppose that there is an $\eps$-sample-\PDP $(\gamma, \gamma')$-robust learner for $\cH$ with excess loss $\alpha$. Then there is an $\eps$-\PDP $(\gamma, \gamma')$-robust learner for with excess loss $O(\alpha / \eps)$ (where the sample complexity remains the same).
\end{lemma}

\begin{proof}
Let $\cA$ denote the $\eps$-sample-\PDP $(\gamma, \gamma')$-robust learner for $\cH$ with sample complexity $m$. The algorithm $\cA'$ draws $m$ samples $(x_1, y_1), \dots, (x_m, y_m)$, and then applies the $\eps$-DP randomized response (see e.g.,~\cite{KasiviswanathanLNRS11}) to each label $y_i$ to get a private label $\ty_i$ and then run $\cA$ on $(x_1, \ty_1), \dots, (x_m, \ty_m)$. The fact that $\cA'$ is $\eps$-\PDP is immediate.

To see its utility guarantee, note that $(x_1, \ty_1), \dots, (x_m, \ty_m)$ are in fact drawn from a distribution $\cD'$ whose probability mass function is
\begin{align*}
\cD'(x, y) = \frac{e^{\eps}}{e^{\eps} + 1} \cD(x, y) + \frac{1}{e^{\eps} + 1} \cD(x, 1 - y).
\end{align*}
The guarantee of $\cA$ implies that
\begin{align*}
\E[\cR_{\gamma'}(h^{priv}; \cD')] &\leq \inf_{h \in \cH} \cR_\gamma(h; \cD') + \alpha.
\end{align*}
Using the definition of $\cD'$, this is exactly equivalent to
\begin{align*}
\frac{1}{e^{\eps} + 1} + \frac{e^{\eps} - 1}{e^{\eps} + 1} \cdot \E[\cR_{\gamma'}(h^{priv}; \cD)] \leq \frac{1}{e^{\eps} + 1} + \frac{e^{\eps} - 1}{e^{\eps} + 1} \cdot \inf_{h \in \cH} \cR_\gamma(h, \cD) + \alpha,
\end{align*}
which in turn is equivalent to $\E[\cR_{\gamma'}(h^{priv}; \cD)] 
\leq \inf_{h \in \cH} \cR_\gamma(h, \cD) + O(\alpha / \eps)$ as desired.
\end{proof} 

The above lemma essentially means that we can focus our attention to sample-\PDP learners for the rest of the section.

\subsection{Robust Empirical Risk Minimization}

For a set $S$ of labeled examples, we write $\cR_\gamma(h, S)$ to denote the $\gamma$-robust error w.r.t. the uniform distribution on $S$ (aka the empirical $\gamma$-robust error). Below we show that, by using a ``smoothened'' version of the loss similar to the hinge loss, we can get the following sample-\PDP ERM algorithm:

\begin{lemma}[ERM for Robust Error] \label{lem:erm-robust}
Let $\gamma' < \gamma$. For any finite hypothesis class $\cH$, there exists an $\eps$-partial DP algorithm that outputs a hypothesis $h^{priv}$ such that
\begin{align*}
\E[\cR_{\gamma'}(h^{priv}; S)] \leq \inf_{h \in \cH} \cR_\gamma(h; S) + O\left(\frac{\log |\cH|}{\eps(\gamma - \gamma')d \cdot |S|}\right).
\end{align*}
\end{lemma}

\begin{proof}
For every $(x, y)$ and $h$, we let $\dec(x, y, h)$ to denote the distance from $x$ to the closest point whose label is not equal to $y$ (i.e., the distance from $x$ to the decision boundary); more formally,
\begin{align*}
\dec(x, y, h) := \min_{z, h(z) \ne y} \|x - z\|_0.
\end{align*}

We define the loss by
\begin{align*}
\ell(h, (x, y)) := \clip_{[0, 1]}\left(\frac{\gamma d - \dec(x, y, h)}{(\gamma - \gamma')d}\right).
\end{align*}
Observe that $\ell$'s sensitivity is at most $\Delta := \frac{1}{(\gamma - \gamma')d}$. By running the exponential mechanism~\cite{McSherryT07}, we obtain $h^{priv}$ such that
\begin{align*}
\E[|S| \cdot \cL(h^{priv}; S)] \leq \min_{h \in \cH} |S| \cdot  \cL(h; S) + O\left(\frac{\log |\cH|}{\eps(\gamma - \gamma')d}\right).
\end{align*}
By dividing the inequality on both sides by $|S|$ and noticing that $\cR_{\gamma'}(h; S) \leq \cL(h; S) \leq \cR_\gamma(h; S)$, we arrive at the desired bound.
\end{proof}

\subsection{Robust Learning of Halfspaces: Reduction to Nets}

We now turn our attention back to halfspaces. A first step is to notice that it suffices consider only halfspaces $\bw$ where $\bw$ belongs to some net. For $\nu > 0$, let $N(\nu)$ denote any $\nu$-net $N$ (under $\ell_1$ metric) of the unit $\ell_1$-ball and let $\cH_{\hs}^{N(\nu)} := \{h_{\bw} \mid \bw \in N\}$. Our formal reduction is stated below.

\begin{lemma} \label{lem:halfspace-net}
Suppose that there exists a $\eps$-sample-\PDP $(\gamma' + \nu, \gamma')$-robust learner for $\cH_{\hs}^{N(\nu)}$ with excess loss $\alpha$. Then, there is also an $\eps$-sample-\PDP $(\gamma' + 3\nu, \gamma')$-robust learner for $\cH_{\hs}$ with excess loss $\alpha$ (where the sample complexity remains the same).
\end{lemma}

\begin{proof}
This follows almost immediately from a claim that
\begin{align} \label{eq:discretizing-margin}
\min_{h \in \cH^{N(\nu)}_{\hs}} \cR_{\gamma^0}(h; \cD) \leq \min_{h \in \cH_{\hs}} \cR_{\gamma^0 + 2\nu}(h; \cD)
\end{align}
for any $\gamma^0 > 0$. To see that this is true, suppose that $h_{\bw^*} := \argmin_{h \in \cH_{\hs}} \cR_{\gamma^0 + \nu}(h; \cD)$. We may rescale $\bw^*$ so that $\|\bw^*\|_1 = 1$. Then, let $\bw'$ be the closest point in $N(\nu)$ to $\bw^*$; by our choice of $\bw^*$, we have $\|\bw - \bw^*\|_1 \leq \nu$. Now, suppose that an example $(x, y)$ is $(\gamma^0 + 2\nu)$-robustly classified by $\bw^*$. We will show that it is $\gamma^0$-robustly classified by $\bw'$, which implies the claim. Assume w.l.o.g. that $y = 1$. Suppose for the sake of contradiction that there exists $z \in B_{\gamma^0 d}(x)$ for which $\left<\bw', z\right> < 0$. Since $\|\bw - \bw^*\|_1 \leq \nu$ and $\|z\|_{\infty} = 1$, we have $\left<\bw^*, z\right> < \nu$. From $\|\bw^*\|_1 = 1$, it is always possible to find $z' \in B_{2\nu d}(z)$ such that $\left<\bw^*, z'\right> < 0$, but $z \in B_{(\gamma^0 + 2\nu)d}(x)$ and therefore contradicts with the assumption that $(x, y)$ is $(\gamma^0 + 2\nu)$-robustly classified.

From the above claim, by simply running the learner for $\cH_{\hs}^{N(\nu)}$, it outputs $h^{priv}$ such that
\begin{align*}
\E[\cR_{\gamma'}(h^{priv}; \cD)] \leq \min_{h \in \cH^{N(\nu)}_{\hs}} \cR_{\gamma' + \nu}(h; \cD) + \alpha \overset{\eqref{eq:discretizing-margin}}{\leq} \min_{h \in \cH_{\hs}} \cR_{\gamma' + 3\nu}(h; \cD) + \alpha.
\qquad\qquad
\qedhere
\end{align*} 
\end{proof}

\subsubsection{Private Robust Learner for Halfspaces}

We start by providing a private learner for $\cH_{\hs}^{N(\nu)}$.

\begin{lemma}
There is an $\eps$-sample-\PDP $(\gamma' + 2\nu, \gamma')$-robust learner for $\cH_{\hs}^{N(\nu)}$ with sample complexity $O\left(\frac{1}{\alpha^2 \nu^2} + \frac{\log(1/\nu)}{\alpha  \eps \cdot \nu}\right)$.
\end{lemma}

\begin{proof}
We run the ERM algorithm from \Cref{lem:erm-robust}, which gives $h^{priv}$ such that
\begin{align*}
\E[\cR_{\gamma'}(h^{priv}; S)]
&\leq \min_{h \in \cH^{N(\nu)}_{\hs}} \cR_{\gamma' + \nu}(h; S) + O\left(\frac{\log |\cH^{N(\nu)}_{\hs}|}{\eps \nu d \cdot n}\right) \\
&= \min_{h \in \cH^{N(\nu)}_{\hs}} \cR_{\gamma' + \nu}(h; S) + O\left(\frac{\log(1/\nu)}{\eps \nu \cdot n}\right).
\end{align*}
From standard generalization bounds~\cite{BartlettM02,McAllester03}, we also have 
\begin{align*}
\max_{h \in \cH^{N(\nu)}_{\hs}} |\cR_{\gamma' + \nu}(h; S) - \cR_{\gamma' + 2\nu}(h; S)| \leq O\left(\frac{1}{\nu^2 \sqrt{n}}\right).
\end{align*}
Combining the above three inequalities, we arrive at the desired bound:
\begin{align*}
\cR_{\gamma'}(h^{priv}; S) \leq \min_{h \in \cH^{N(\nu)}_{\hs}} \cR_{\gamma' + 2\nu}(h; \cD) + O\left(\frac{1}{\nu^2 \sqrt{n}} + \frac{\log(1/\nu)}{\eps \nu \cdot n}\right).
\qquad\qquad\qedhere
\end{align*}
\end{proof}

Combining the above lemma with \Cref{lem:halfspace-net} with $\nu = (\gamma - \gamma')/5$, we arrive at:

\begin{lemma} \label{lem:robust-halfspaces-sample-only}
There is an $\eps$-sample-\PDP $(\gamma, \gamma')$-robust learner for halfspaces with excess error $$O\left(\frac{1}{(\gamma - \gamma')^2 \sqrt{n}} + \frac{\log(1/(\gamma - \gamma'))}{\eps (\gamma - \gamma') \cdot n}\right).$$
\end{lemma}

Combining the above lemma with \Cref{lem:label-rr}, we get \Cref{lem:robust-halfspaces}.

\section{Discussion}
\label{sec:disc}

In this section, we discuss the meaning of partial DP.
Ideally, of course, we would provide a standard DP guarantee with a small privacy loss bound (say, $(\varepsilon,\delta)$-DP with $\varepsilon\le1$ and $\delta \le 10^{-6}$).
However, in practice, we are seeing large privacy loss bounds ($\varepsilon\ge10$) and we lack a satisfactory way to interpret such guarantees.

Thus the premise of this discussion is that we are in a setting where, in order to provide reasonable utility, we need a large $\varepsilon$ under the standard definition of $(\varepsilon,\delta)$-DP. The fundamental question is: \emph{How can we justify $(\varepsilon,\delta)$-DP with large $\varepsilon$?} And, even more importantly, when can we \emph{not} justify this? 

Intuitively, large $\varepsilon$s can be justified by informally arguing DP is a worst-case definition and this worst case is not realistic.
The goal of partial DP is to provide a framework for formalizing this intuition for justifying large $\varepsilon$ which is precise enough to also fail to justify large $\varepsilon$ when the algorithm at hand does indeed exhibit worst-case behaviour.
For example, if $\varepsilon(x,x')$ is large when the only difference between $x$ and $x'$ is that the person visited a given website, then we clearly do not have a meaningful privacy guarantee.

\para{Interpretation.}
Partial DP provides a language to formalize the intuitive notion of a ``nice algorithm.'' Specifically, it allows us to rule out algorithms that act like performing randomized response on some sensitive feature. For example, if we want to formalize the constraint that the algorithm does not reveal whether or not a given person has a certain disease, we would require that $\varepsilon(x,x')$ is small whenever the only difference between $x$ and $x'$ is that person's disease status.

To interpret a partial DP guarantee, we must also discuss what constitutes a ``realistic adversary.''
There are many different ways to restrict the adversary (see \S\ref{sec:adv-ass}). Per-attribute partial DP naturally corresponds to assuming that the adversary is interested in learning a function of only a few attributes, whereas standard DP protects an arbitrary function of all the attributes of a person.\footnote{The quantitative guarantee will degrade gracefully with the number of attributes the adversary is interested in.}
E.g., if the dataset is employment records, we can assume that the adversary wishes to learn the target's salary, but is not particularly interested in learning their age or whether or not they are an employee. 
Such assumptions can be justified in a variety of ways, depending on context. In the prior example, age may already be public information and the employer may be willing to disclose who is or is not an employee.
In general, the interpretation of partial DP is context-dependent; the effectiveness of the guarantee depends on what kind of information leakage is concerning. 

\para{Limitations.}
It is also important to discuss the attacks that partial DP does \emph{not} protect against.\footnote{Partial DP implies some standard DP guarantee and hence protects against arbitrary adversaries, but the privacy parameter may be large.}
Membership inference attacks \citep{ShokriSSS17,DworkSSUV15} are an example of a worst-case attack --- whether a person is included in the the dataset or excluded is a function of all the attributes and hence partial DP does not provide a better guarantee than standard DP. Whether membership of the dataset is sensitive depends on the context.
For example, if the dataset consists of the medical records of patients with a certain medical condition, then a membership inference attack can reveal that the target of the attack has that medical condition. In this case, per-attribte partial DP is not particularly useful.
On the other hand, if the dataset consists of all people with public non-anonymous profiles on a social media network, then membership in this dataset is likely not sensitive.

\para{Correlated Attributes.}
Sensitive information may be repeated across multiple attributes and, in this case, the guarantee of per-attribute partial DP would rapidly degrade. E.g., if each attribute is a person's location at a given time, then their home address will be repeated across many attributes; hence such time series data is a bad use case for per-attribute partial DP.

Some attributes will be loosely correlated -- e.g., age and wealth -- but we do not consider this to be a problem if the relationship is not strict. 
While it is possible to, say, guess the income of a person based on their demographic information, this is generally \emph{not} considered to be a privacy violation \cite{frankblog2016,DPorg-inference-is-not-a-privacy-violation}.

In general, privacy should be thought of in terms of causal relationships, not statistical correlations \cite{tschantz2017differential}. Indeed the definition of DP is precisely a causal property, as it considers a pair of datasets, which correspond to the real dataset and a hypothetical counterfactual dataset. The definition of DP does not make any distributional assumptions about the data.

Correlations are present not only between the attributes of a single person, but also between the data of different people. For example, whether or not a given person has an infectious disesase is highly correlated with whether or not the people around them have that disease. Thus revealing the fact that there is an outbreak of an infectious disease reveals information about specific individuals. But this is not a privacy violation. (And if we treated this as a privacy violation, it would prevent us from revealing useful information about disease outbreaks.) However, revealing a specific person's test result is a potential privacy violation -- the key is that there is a direct causal relationship between someone's data (i.e., their test result) and the information being released. 
By the same token, a person may have many attributes that are correlated with having a certain disease, but releasing those correlated attributes is fundamentally different from releasing an actual diagnosis.

\para{Concrete Example: 2020 US Census.}\label{par:census}
The redistricting data from the 2020 US Census was released in a DP manner. The generally quoted guarantee is $(17.14,10^{-10})$-DP plus $(2.47,10^{-10})$-DP \cite{abowd20222020} and applying basic composition to these two releases gives $\varepsilon = 19.61$.
To be more precise, the redistricting data satisfies $2.56+0.07$-zCDP, which implies $(13.8,10^{-6})$-DP.

The Census Bureau provide provide a detailed breakdown of the privacy allocation \cite{abowd20222020,census2021breakdown,censusparams2021}. This exactly corresponds to a partial CDP guarantee. Their TopDown algorithm computes multiple histograms across subsets of attributes (which bears some similarity to our heavy hitters algorithm in Section \ref{sec:heavy_hitters}). To determine the $\varepsilon$-$\nabla$CDP guarantee, for any $x,x'\in\mathcal{X}$, determine the set of attributes on which they differ and then we look at which histograms involve those attributes to determine $\varepsilon(x,x')$. Histograms that only involve attributes on which $x$ and $x'$ agree need not be accounted for under partial DP.

To make this partial CDP guarantee concrete, we can relate it to a specific attack. The Census Bureau conducted a simulated reconstruction and reidentification experiment \cite{desfontainesblog20210526}. The punchline of their simulated attack was learning people’s race and ethnicity from the data that was publicly released from the 2010 US Census. We can calculate a privacy guarantee for just these two attributes in the 2020 release. Specifically, if we allow a person’s race and ethnicity to change, but their other attributes are fixed, then we get a $1.02$-zCDP guarantee, which yields $(7.85,10^{-6})$-DP. That is, in terms of partial DP $\frac12 \varepsilon(x,x')^2 \le 1.02$ when $x$ and $x'$ differ only on the race and ethnicity fields.

\para{Generality of Partial DP Definition.}
Our algorithmic results (in Sections \ref{sec:gen_alg_results}, \ref{sec:heavy_hitters}, \& \ref{sec:robust-halfspaces}) focus on per-attribute partial DP. For simplicity, we assume each attribute has the same privacy parameter $\varepsilon_0$. But not all attributes will be equally sensitive. Thus it is natural to to consider per-attribute guarantees where each attribute has its own privacy parameters. In the 2020 US Census example, we see that the attributes have different privacy parameters.
Our definition of partial DP is general enough to capture such non-uniform per-attribute privacy guarantees. 

As mentioned in Section \ref{sec:related-work}, it is possible to give an even more general definition than we do, where the $\varepsilon$ metric considers a pair of datasets, not just a pair of individual records. 
While such added generality may seem like a feature, it makes such a definition harder to interpret. In particular, it becomes hard to relate such a definition back to standard DP. Thus we deliberately choose not to define partial DP so generally.

\subsection{Assumptions about the Adversary}\label{sec:adv-ass}

To give refined privacy guarantees (including, but not limited to, partial DP) meaning, we must give a characterization of what we might consider reasonable restricted adversaries, which we can then use to interpret our definition. In this section, we discuss different types of restricted adversaries.

The definition of DP does not explicitly mention an adversary; it simply states that the output distribution of the algorithm does not change much (as measured by $\varepsilon$) if we arbitrarily change the data of one individual in the input of the algorithm.
However, to interpret this definition and give meaning to the privacy guarantee, we must envisage an adversary who sees the output of the algorithm, combines this information with their knowledge, and thereby potentially learns a piece of information about an individual that they should not have been able to learn. The adversary could be a stranger, a close friend or relative, a government entity, or a private entity we do business with and each of these potential adversaries will have different knowledge, resources, and goals.

In effect, standard DP makes minimal assumptions about the adversary---the adversary can have near-complete knowledge of the dataset and can target an arbitrary piece of information about an arbitrary individual.

There are four ways in which we could make assumptions that constrain the adversary:

    \textbf{(i) Assumptions about the Adversary's Knowledge.} 
    DP effectively permits the adversary to know everything about the dataset except for the one bit of private information that they are seeking to extract. Although the adversary may have access to a lot of information from auxiliary data sources, it is unrealistic to assume that this information is so complete and so accurate. Thus it is natural to assume some uncertainty about the dataset in the eyes of the adversary; this could be formalized by endowing the dataset with randomness and exploiting this randomness in the privacy guarantee \cite{BhaskarBGLT11,bassily2013coupled,bhowmick2018protection}.
    
    However, such assumptions about the adversary's knowledge are very brittle \cite{DPorg-average-case-dp,DPorg-privacy-composition}. In particular, assumptions about the adversary's knowledge can be invalidated by \emph{future} releases of information. That is, we may assume that certain information is unknown to the adversary, but subsequently an auxiliary dataset is made available that contains this information; when that happens, it is too late to retract the output of our algorithm. Such assumptions are also not robust to composition. That is, the output of our algorithm may itself invalidate these assumptions, so, if we run another algorithm subsequently, we cannot make the same assumptions again.  
    
    It is also difficult to effectively formulate such assumptions about the adversary's knowledge. For example, it is tempting to assume that the data consists of i.i.d.~samples from some nice distribution. However, this corresponds to assuming an entirely na\"ive adversary with effectively no knowledge of the dataset beyond the general characteristics of the population it is collected from.

    \textbf{(ii) Assumptions about the Target Individual or Dataset.}
    Distributional assumptions about the data can also encode a different type of privacy guarantee (as opposed to that distribution representing the uncertainty of the adversary). 
    Intuitively, we can encode the assumption that the adversary only targets ``typical'' individuals in ``typical'' datasets and the privacy guarantee may fail for individual outliers or abnormal datasets.
    For example, a basic DP algorithm is to add noise to some statistic that is scaled to its sensitivity; an average-case assumption about the target individual and dataset would allow us to replace the worst-case sensitivity with a notion of average-case sensitivity \cite{HallRW11,TriastcynF20}. However, an individual deviating significantly from the rest of the dataset will have a correspondingly weaker privacy guarantee \cite{Wang19}. 
    
    This approach has two deficiencies: 
    First, providing unequal privacy protection raises ethical questions. 
    Second, it may be unnecessary to make this compromise. For example, techniques such as clipping can control the worst-case sensitivity or we can use smooth sensitivity \cite{NissimRS07}. We can also test whether the dataset is typical before performing the analysis and abort if it is not \cite{DworkL09}. Thus it is often possible to obtain the benefits of average-case assumptions on the data while still attaining standard DP.

    \textbf{(iii) Assumptions about the Adversary's Capabilities.}
    Performing a privacy attack generally requires effort. Thus we may make assumptions about the adversary's ability or willingness to perform the attack \cite{chaudhuri2019capacity}. A good example is computational DP~\cite{MironovPRV09}, where we assume that the adversary's computational power is limited and thus they cannot, for example, break a cryptographic system.
    
    We can also assume that the adversary will only perform certain types of attacks. For example, $k$-anonymity and related definitions are tailored to preventing a specific style of record-linkage attacks. In the same vein, the data curator can simulate an attack on the output of their algorithm \cite{CarliniLEKS19,JagielskiUO20,carlini2022membership}. On one hand, the success of the simulated attack establishes that the algorithm is not DP---and this can be used to check the privacy analysis \cite{tramer2022debugging}. On the other hand, the failure of the simulated attack establishes a privacy guarantee that is meaningful as long as the real adversary is similar to the simulated adversary.

    \textbf{(iv) Assumptions about the Adversary's Goals.}
    DP protects against an adversary seeking to learn an arbitrary one-bit function of the target individual's data. Equivalently (up to a factor of two in the privacy parameter), it prevents the adversary from learning whether or not the target individual's data was included in the dataset. While being included in the dataset may be sensitive depending on how the dataset was collected \cite{ShokriSSS17,DworkSSUV15}, this often does not correspond to a realistic threat.
    
    Thus we can relax the definition to protect only certain pieces of information from attacks. This corresponds to making an assumption about what function the adversary wants to learn about the target individual. 
    For example, if the dataset corresponds to employment records, we can assume that the adversary wishes to learn the target's salary or their performance rating, but is not particularly interested in learning their age or whether or not they are an employee. 
    Such assumptions can be justified in a variety of ways. In the prior example, it may be the case that age is already public information and that the employer is willing to disclose who is or is not an employee.

Our partial DP approach corresponds to making assumptions about the adversary's goals. Specifically, we provide guarantees for adversaries whose goal is to learn a single attribute or a function of a few attributes. %
Whether such an assumption corresponds to a realistic adversary will depend on the application domain.

\paragraph{Comparison to Prior Approaches.}
As discussed in Section \ref{sec:related-work}, there has been a lot of prior exploration of privacy definitions; and much of it is very similar to our definition.
Yet, despite this exploratory work, there has been very little adoption of these relaxations of DP in either theory or practice.\footnote{We remark that quantitative relaxations of DP, such as Concentrated DP \cite{DworkR16,BunS16}, R\'enyi DP \cite{Mironov17}, and Gaussian DP \cite{dong2019gaussian}, have seen widespread adoption. However, this is orthogonal to our work; these definitions change how we measure closeness of distributions, whereas our work changes which distributions we compare. Similarly, different trust models, such as local DP \cite{KasiviswanathanLNRS11}, shuffled DP \cite{cheu2019distributed,erlingsson2019amplification}, multi-central DP \cite{steinke2020multi}, and pan-privacy \cite{dwork2010pan}, have been explored and/or deployed, but these are also orthogonal; they change how the adversary interacts with the system, rather than the final privacy guarantee.}
Why is this? And what makes our work different?

Although we cannot say with certainty why there has been a limited adoption of these alternative versions of DP, we believe that most prior works either do not (i) give a convincing interpretation of their privacy definition or they do not (ii) demonstrate that the new definition opens up sufficiently interesting algorithmic applications (or both). We now reiterate how our work addresses these questions.

\textbf{(i)} In Section \ref{sec:adv-ass}, we argue that a privacy definition should be interpreted in terms of the adversaries it protects against, and we argue that partial DP can be interpreted as providing stronger privacy protections against adversaries that are only interested in part of a person's record, such as a function of only a few attributes.
We emphasize again that partial DP guarantees should be viewed as complementary the standard notion of DP and not a replacement.  %
Thus partial DP does not represent a radical departure from the established definition of DP. The value of partial DP is in the setting where we are pushing the limits of DP---i.e., the privacy loss bound $\varepsilon$ of the standard DP definition is uncomfortably large (but not crazy). Partial DP provides a formalism to ``break down'' the $\varepsilon$ parameter and tie it to a specific adversary or attack.

\textbf{(ii)} In Sections \ref{sec:gen_alg_results}, \ref{sec:heavy_hitters}, and \ref{sec:robust-halfspaces}, we provide a variety of algorithmic results that fit the (per-attribute) partial DP definition, which demonstrates that partial DP opens up interesting algorithmic questions.
There are two salient points in these results: First, they show that there is a separation between the per-attribute partial DP parameter and the standard DP parameter, which implies that the interpretation of partial DP could be meaningful. Second, since the algorithms we analyze are quite diverse, they demonstrate that our definition is not ``overfitted'' to one particular algorithm. 

\paragraph{Bayesian \& Information Theoretic Interpretations.}
Applying Bayes' law to the standard DP definition gives us a ``semantic'' interpretation of the guarantee---regardless of the adversary's prior beliefs, after seeing the DP output, their posterior beliefs cannot change much based on a single person's data \cite{kasiviswanathan2014semantics}. 
We can give a similar interpretation for pure partial DP:
\begin{proposition}\label{prop:semantic}
    Let $\mathcal{X} = \mathcal{X}_1 \times \cdots \times \mathcal{X}_d$. 
    Let $M : \mathcal{X}^n \to \mathcal{Y}$ satisfy $\varepsilon$-$\nabla_0$DP.
    Let $P$ be a distribution on $\mathcal{X}^n$ representing the adversary's prior beliefs. %
    Assume that for some $s = \{j_1, \cdots, j_{|s|}\} \subset [d]$, the distribution $P$ can be decomposed as a product distribution over the attributes given by $s$ and the remaining $d-|s|$ attributes.
    Let $P_{X|M(X)=y}$ denote the conditional distribution of $X$ obtained by drawing $X \gets P$ and conditioning on the event $M(X)=y$ for some fixed $y \in \mathcal{Y}$. 
    Similarly, let $P_{X|M(X_{-i})=y}$ denote the conditional distribution of $X$ obtained by drawing $X \gets P$ and conditioning on the event $M(X_{-i})=y$ for some fixed $y \in \mathcal{Y}$ and $i \in [n]$, where $X_{-i}$ denotes $X$ with the $i$th entry removed or blanked. 
    Let $P_{X_s}$ denote the marginal distribution on $\mathcal{X}_{j_1} \times \cdots \times \mathcal{X}_{j_{|s|}}$ obtained by sampling $X \gets P$ and only revealing the attributes indexed by $s$. Define $P_{X_s|M(X)=y}$ and $P_{X_s|M(X_{-i})=y}$ analogously as marginals of the conditional distributions.
    Then, for all $y \in \mathcal{Y}$ and all $i \in [n]$, \[d_{\mathrm{TV}}(P_{X_s|M(X)=y},P_{X_s|M(X_{-i})=y}) \le e^{2|s|\varepsilon}-1.\]
\end{proposition}
The proof of Proposition \ref{prop:semantic} directly follows that of \cite{kasiviswanathan2014semantics}. The assumptions of the proposition allow us to essentially ignore all the attributes indexed by $[d]\setminus s$ and, once we discard those irrelevant attributes, we have a $(|s|\varepsilon)$-DP algorithm.

We remark that the product distribution assumption may seem strong, as it implies that there is no connection between the attributes in $s$ and the attributes in $[d] \setminus s$. However, there is one very simple way that this can arise: Suppose the adversary already knows the value of all of the attributes in $[d] \setminus s$. In this case we have a trivial product distribution where the distribution on the attributes in $[d] \setminus s$ is a point mass.

We can also give an information-theoretic interpretation:
Suppose $X_1, \cdots, X_n \in \mathcal{X} = \mathcal{X}_1 \times \cdots \times \mathcal{X}_d$ are independent random variables.
These correspond to the data of $n$ individuals. The independence assumption is essentially saying that the adversary knows the population, but knows neither the individuals nor the relationships between the individuals.
If $M$ is $\varepsilon$-DP or $\frac12\varepsilon^2$-zCDP, then $I(X_i;M(X)) \le \frac12 \varepsilon^2$ for all $i \in [n]$ \cite{BunS16}. Here $I(\cdot;\cdot)$ denotes the mutual information in nats.
That is to say, DP bounds the amount of information that $M$ reveals about any given record in the dataset.
Such a bound can be meaningful even in the large $\varepsilon$ regime; if the individual's data is high-entropy, then we cannot reconstruct it, even if we can learn some of it \cite{bhowmick2018protection}. 
We can give a stronger guarantee under partial DP:
\begin{proposition}
    Suppose $X_1, \ldots, X_n \in \mathcal{X} = \mathcal{X}_1 \times \cdots \times \mathcal{X}_d$ are independent random variables. Fix $i \in [n]$.
    Suppose we can partition $[d]= s_1 \cup \cdots \cup s_k$ such that $X_{i,s_1}, \dots, X_{i,s_k}$ are independent random variables, where $X_{i,s_j}$ denotes the attributes indexed by $s_j$ of individual $i$.
    Let $M : \mathcal{X}^n \to \mathcal{Y}$ satisfy $\varepsilon$-$\nabla_0$CDP. Then $I(X_{i,s_j};M(X)) \le \frac12 \varepsilon^2 |s_j|^2$ for all $j \in [k]$ and $I(X_i,M(X)) \le \frac12 \varepsilon^2 \sum_{j=1}^k |s_j|^2$.
\end{proposition}

\paragraph{Applicability of Partial DP Depends on Context.}
Whether or not partial DP guarantees are meaningful (beyond the implied standard DP guarantee) will depend on the context. In particular, it depends on what kind of adversaries we need to protect against.
This limitation is inherent---if we want a context-independent guarantee of individual privacy, we cannot do better than the standard DP definition.

The limitations of (per-attribe) partial DP are the attacks that it does \emph{not} give good protection against (beyond the baseline standard DP guarantee that partial DP implies). 
As we have mentioned, partial DP does not provide enhanced protection against membership inference attacks. Thus, if the data selection process itself reveals sensitive information, partial DP is not helpful. For example, if the dataset consists of only HIV patients, then membership inference can reveal that a participant is HIV positive. 
On the other hand, if the dataset consists of all patients of a given hospital system (or a random sample of those patients), then that is potentially less sensitive and membership inference is less of a concern; this is a setting where partial DP may be meaningful. Another example would be a nationwide census. A census dataset ideally contains everyone in the country, so membership inference would only reveal that someone was living in the country.

To interpret the guarantee of per-attribe partial DP, we must also consider whether sensitive information is reflected in many attributes. For example, suppose each attribute is the person's location at a given point in time. We would expect that the person is at home for extended periods of time, so the home location would be revealed in many attributes simultaneously. Partial DP would not be particularly useful in such a scenario. Other bad use cases for partial DP include the setting where each attribute is a text message or photo and a person can contribute many text messages and photos and sensitive information may be repeated in many of those messages or photos \cite{asi2019element}.  This setting corresponds to ``item-level'' or ``event-level'' DP \cite{levy2021learning}.

We consider the ideal setting for partial DP to be the case where attributes are heterogeneous (as opposed to the settings discussed in the previous paragraph where the attributes are homogeneous locations, messages, or images). 
For example, age, race/ethnicity, gender, home address, income, sexuality, medical status, occupation, criminal history, relationship status, commute length, and immigration status are heterogeneous attributes. These attributes are not independent, but the correlations between them are relatively weak. We can thus hope that partial DP provides meaningful guarantees on a dataset containing these attributes. In particular, it is unlikely that an adversary is interested in some complex function covering all or most of these attributes. A realistic adversary is likely only interested in one of these attributes or maybe a pair of them.

In any case, providing a partial DP guarantee contains more information than the single parameter of the standard DP definition. Thus we argue that, even in settings where partial DP is not particularly appropriate, it is still no worse than simply providing a standard DP guarantee.

\section{Conclusion}\label{sec:conc}

We have presented several algorithms and analyzed them under partial DP, which gives more granular privacy guarantees than standard DP. 
Our results demonstrate that there are multiple separations between the achievable per-attribute $\varepsilon_0$-$\nabla_0$DP and the per-person $\varepsilon$-DP parameters --- i.e., settings where the achievable per-person parameter is large (say, $\varepsilon \ge 10$), but we can still give more granular guarantees with a smaller parameter (e.g., $\varepsilon_0 \le 1$).
In this case the partial DP guarantee with a small $\varepsilon_0$ may be a meaningful and useful (depending on the context) supplement to standard DP guarantee with a large $\varepsilon$, as the smaller parameter can interpreted more easily, albeit on a per-attribute basis.

We note that interpreting the guarantees of partial DP depends on what constitutes a ``realistic'' adversary. Per-attribute partial DP naturally corresponds to assuming that the adversary is interested in learning a function of only a few attributes.  On the flip side, partial DP may not provide better guarantees than standard DP for membership inference attacks~\citep{ShokriSSS17,DworkSSUV15} and in the cases where sensitive information may be repeated across multiple attributes.

We hope that our work inspires further exploration of more granular privacy guarantees, and further expansion of the DP algorithmic toolkit.

\paragraph{Further work.}
We hope that our work inspires further study of refined privacy guarantees. Our partial DP framework can and should be explored further, both in terms of developing and applying algorithms and in terms of further developing the definition. In particular, most of our results are restricted to per-attribute partial DP. A natural extension is to have a different $\varepsilon_i$ for each attribute, as some attributes are more sensitive than others.

A specific open question for algorithms development is to improve Theorem \ref{thm:pdp-mwem} so that we can obtain max error guarantees under per-attribute partial DP (like in the standard DP setting), rather than needing to average over a set of $\ell$ queries.

Going beyond partial DP, there is scope for other refined definitions that capture some limitations on the adversary. We have discussed some possible directions in Section \ref{sec:adv-ass}.
To facilitate such work, we propose the following desiderata for other refined definitions:
\begin{itemize}
    \item The adversary should not be ``baked in'' to the privacy definition. That is, the definition should be stated in a way that can be interpreted and verified without knowing the specifics of the adversary. Our definition, like the original DP definition \cite{DworkMNS06}, is frequentist, rather than Bayesian --- i.e., it does not mention the adversary's beliefs and instead asserts that the output distributions are indistinguishable; this makes it easier to use. 
    \item The privacy definition should not be overly tailored to a specific algorithm. It is important to disentangle algorithms from definitions; this is one of the important conceptual contributions of the original DP definition. In other words, definitions should be re-usable.
    \item The new privacy definition should be formally comparable to the standard definition of DP (as we show in Proposition \ref{prop:ppdp-conversion}) or, at least, it should be possible to ensure that algorithms satisfy both definitions simultaneously. We believe that the goal of such research should not be to replace the standard definition of DP, rather the goal should be to supplement it. 
    \item In order for a new privacy definition to be useful, we must demonstrate a quantitative separation between it and the standard definition of DP, as we have done with our algorithmic results. A new definition should not be a substitute for designing better algorithms. For example, when computing the mean of unbounded Gaussian data we could either devise an average-case privacy definition to avoid dealing with the infinite global sensitivity of the mean, or we could simply clip the data \cite{KarwaV17}; we argue that the second option is vastly preferable. Thus, to justify a new privacy definition, new algorithmic results should be matched to a lower bound showing that it is impossible to match the performance under the usual definition of DP.
\end{itemize}

\nocite{bhowmick2018protection}
\nocite{levy2021learning}
\nocite{bassily2013coupled}
\nocite{chaudhuri2019capacity}
\nocite{frankblog2016}

\addcontentsline{toc}{section}{References}
\printbibliography

\newpage
\onecolumn

\appendix

\section{Missing Proofs from~\Cref{sec:heavy_hitters}}
\subsection{Analysis}
\label{app:hist-proof}

\subsubsection{Privacy Analysis}

In this section, we will prove the privacy guarantee of~\Cref{alg:priv-tree} as summarized below.

\begin{lemma}[Privacy Guarantee of \textsc{PrivHeavyHitter}] \label{lem:priv-tree-privacy}
Suppose that $\lambda, \mu$ be such that $\mu > 1$ and $\frac{2}{\lambda}\left(1 + \frac{1}{1 - e^{-\mu/\lambda}} \right) \leq \eps$, then \textsc{PrivHeavyHitter} is $\eps$-\PDP.
\end{lemma}

In fact, we will prove even outputting all the sets $(L_I)_{I \in \cI}$ is $\eps$-\PDP. For clarity, below we write the frequencies and lists as functions of the input datasets $\bD$ or $\bD'$. In other words, we would like to show that:

\begin{lemma}
Suppose that $\lambda, \mu$ satisfy the conditions in \Cref{lem:priv-tree-privacy}.
For any neighboring datasets $\bD, \bD'$ and any values of $(S_I)_{I \in \cI}$, it holds that $\Pr\left[\forall I \in \cI, L_I(\bD) = S_I\right] \leq e^{\eps} \cdot \Pr\left[\forall I \in \cI, L_I(\bD') = S_I\right].$
\end{lemma}

\begin{proof}
Due to symmetry, it suffices to consider the case where the differing elements in $\bD$ and $\bD'$ are $0_d$ and $10_{d - 1}$ respectively. We may write the LHS probability as
\begin{align*}
\Pr\left[\forall I \in \cI, L_I(\bD) = S_I\right]
&= \prod_{\ell \in [\log d]} \Pr\left[\forall I \in \cI_\ell, L_I(\bD) = S_I \mid \forall I' \in \cI_{\ell - 1}, L_{I'}(\bD) = S_{I'}\right] \\
&= \prod_{\ell \in [\log d]} \prod_{I \in \cI_\ell} \Pr\left[L_I(\bD) = S_I \mid \forall I' \in \cI_{\ell - 1}, L_{I'}(\bD) = S_{I'}\right] \\
&= \prod_{\ell \in [\log d]} \prod_{I \in \cI_\ell} \Pr\left[L_I(\bD) = S_I \mid L_{\Ileft}(\bD) = S_{\Ileft} \mbox{ and } L_{\Iright}(\bD) = S_{\Iright}\right].
\end{align*}
This also means that the probability is zero (for both $\bD$ and $\bD'$) if $S_{I} \nsubseteq S_{\Ileft} \circ S_{\Iright}$ for some $I \in \cI$. Therefore, we may assume henceforth that $S_{I} \subseteq S_{\Ileft} \circ S_{\Iright}$ for all $I \in \cI$. Under this assumption, we may further write the above probability as
\begin{align*}
\Pr\left[\forall I \in \cI, L_I(\bD) = S_I\right]
&= \prod_{\ell \in [\log d]} \prod_{I \in \cI_\ell} \left(\prod_{s \in S_I} \Pr[\hf_s^I(\bD) > \tau_\ell]\right) \left(\prod_{s \in (S_{\Ileft} \circ S_{\Iright}) \setminus S_I} \Pr[\hf_s(\bD) \leq \tau_\ell]\right).
\end{align*}
We let $I^\ell := \{1, \dots, 2^{\ell}\}$. Notice here that, when $s \ne 0_{2^\ell}, 10_{2^{\ell - 1}}$ or $I \ne I^\ell$, we have $f_s^I(\bD) = f_s^I(\bD')$. This means that
\begin{align*}
&\frac{\Pr\left[\forall I \in \cI, L_I(\bD) = S_I\right]}{\Pr\left[\forall I \in \cI, L_I(\bD') = S_I\right]} \\
&= \prod_{\ell \in [\log d]} \prod_{I \in \cI_\ell} \left(\prod_{s \in S_I} \frac{\Pr[\hf_s^I(\bD) > \tau_\ell]}{\Pr[\hf_s^I(\bD') > \tau_\ell]}\right) \left(\prod_{s \in (S_{\Ileft} \circ S_{\Iright}) \setminus S_I} \frac{\Pr[\hf_s^I(\bD) \leq \tau_\ell]}{\Pr[\hf_s^I(\bD') \leq \tau_\ell]}\right) \\
&= \prod_{\ell \in [\log d]} \left(\prod_{s \in S_{I^\ell} \atop s \in \left\{0_{2^\ell}, 10_{2^{\ell - 1}}\right\}} \frac{\Pr[\hf_s^{I^\ell}(\bD) > \tau_\ell]}{\Pr[\hf_s^{I^\ell}(\bD') > \tau_\ell]}\right) \left(\prod_{s \in \left(S_{\Ileft^\ell} \circ S_{\Iright^\ell}\right) \setminus S_{I^\ell} \atop s \in \left\{0_{2^\ell}, 10_{2^{\ell - 1}}\right\}} \frac{\Pr[\hf_s^{I^\ell}(\bD) \leq \tau_\ell]}{\Pr[\hf_s^{I^\ell}(\bD') \leq \tau_\ell]}\right). \\
\end{align*}
From now on, when we write $\hf_s$ and $f_s$ we always mean $\hf_s^{I^\ell}$ and $f_s^{I^\ell}$ and where $\ell$ is the length of $s$, respectively. Therefore, we will drop the superscipt $I^\ell$ for notational ease.

We will separately show that
\begin{align} \label{eq:above-threshold-prob}
\prod_{\ell \in [\log d]} \left(\prod_{s \in S_{I^\ell} \atop s \in \left\{0_{2^\ell}, 10_{2^{\ell - 1}}\right\}} \frac{\Pr[\hf_s(\bD) > \tau_\ell]}{\Pr[\hf_s(\bD') > \tau_\ell]}\right) \leq e^{\eps/2},
\end{align}
and
\begin{align} \label{eq:below-threshold-prob}
\prod_{\ell \in [\log d]} \left(\prod_{s \in \left(S_{\Ileft^\ell} \circ S_{\Iright^\ell}\right) \setminus S_{I^\ell} \atop s \in \left\{0_{2^\ell}, 10_{2^{\ell - 1}}\right\}} \frac{\Pr[\hf_s(\bD) \leq \tau_\ell]}{\Pr[\hf_s(\bD') \leq \tau_\ell]}\right) \leq e^{\eps/2}.
\end{align}
Multiplying the two together yields the desired bound.

We now prove~\eqref{eq:above-threshold-prob}. Since $f_{01_{2^\ell - 1}}(\bD) < f_{01_{2^\ell - 1}}(\bD')$, we have $\Pr[\hf_{01_{2^\ell - 1}}(\bD) > \tau_\ell] \leq \Pr[\hf_{01_{2^\ell - 1}}(\bD') > \tau_\ell]$. Therefore,
\begin{align}
\prod_{\ell \in [\log d]} \left(\prod_{s \in S_{I^\ell} \atop s \in \left\{0_{2^\ell}, 10_{2^{\ell - 1}}\right\}} \frac{\Pr[\hf_s(\bD) > \tau_\ell]}{\Pr[\hf_s(\bD') > \tau_\ell]}\right)
&\leq \prod_{\ell \in [\log d]} \left(\prod_{s \in S_{I^\ell} \atop s = 0_{2^\ell}} \frac{\Pr[\hf_s(\bD) > \tau_\ell]}{\Pr[\hf_s(\bD') > \tau_\ell]}\right) \nonumber \\
&= \prod_{\ell \in [\log d] \atop 0_{2^\ell} \in S^{{I^\ell}}_\ell} \frac{\Pr[\hf_{0_{2^\ell}}(\bD) > \tau_\ell]}{\Pr[\hf_{0_{2^\ell}}(\bD') > \tau_\ell]}.
\end{align}
Now, let $L_0$ denote the smallest integer such that $f_{0_{2^{L_0}}}(\bD) > \tau_{L_0} - \mu$. First, notice that for all $\ell < L_0$ we have $\min\{\hf_{0_{2^\ell}}(\bD), \tau_\ell - \mu\} = \tau_\ell - \mu = \min\{\hf_{0_{2^\ell}}(\bD'), \tau_\ell - \mu\}$. This means that
\begin{align*}
\forall \ell < L_0, \frac{\Pr[\hf_{0_{2^\ell}}(\bD) > \tau_\ell]}{\Pr[\hf_{0_{2^\ell}}(\bD') > \tau_\ell]} = 1.
\end{align*}
Next, consider $\ell = L_0$. In this case, notice that $\min\{\hf_{0_{2^\ell}}(\bD), \tau_\ell - \mu\} - \min\{\hf_{0_{2^\ell}}(\bD'), \tau_\ell - \mu\} \leq 1$. As a result, by standard DP property of the Laplace mechanism, we have
\begin{align*}
\frac{\Pr[\hf_{0_{2^{L_0}}}(\bD) > \tau_{L_0}]}{\Pr[\hf_{0_{2^{L_0}}}(\bD') > \tau_{L_0}]} \leq e^{1/\lambda}.
\end{align*}
Finally, for $\ell > L_0$, we have $f_{0_{2^{\ell}}}(\bD') > \tau_\ell + (\ell - L_0)\mu - 1$. This implies that
\begin{align*}
\frac{\Pr[\hf_{0_{2^{\ell}}}(\bD) > \tau_{\ell}]}{\Pr[\hf_{0_{2^{\ell}}}(\bD') > \tau_{\ell}]}
&\leq \frac{\Pr[(\ell - L_0)\mu + \Lap(\lambda) > 0]}{\Pr[(\ell - L_0)\mu - 1 + \Lap(\lambda) > 0]} \\
&\leq \exp\left(\frac{1}{\lambda} \cdot \exp\left(\frac{1 - (\ell - L_0)\mu}{\lambda}\right)\right),
\end{align*}
where the last inequality follows from Lemma 2.1 of~\cite{ZhangXX16}.

Combining these inequalities, we have
\begin{align*}
\prod_{\ell \in [\log d]} \left(\prod_{s \in S_{I^\ell} \atop s \in \left\{0_{2^\ell}, 10_{2^{\ell - 1}}\right\}} \frac{\Pr[\hf_s(\bD) > \tau_\ell]}{\Pr[\hf_s(\bD') > \tau_\ell]}\right) %
&\leq e^{1/\lambda} \cdot \prod_{i=1}^{\infty} \exp\left(\frac{1}{\lambda} \cdot \exp\left(\frac{1 - i \mu}{\lambda}\right)\right) \\
&\leq e^{1/\lambda} \cdot \prod_{i=1}^{\infty} \exp\left(\frac{1}{\lambda} \cdot \exp\left(\frac{-(i - 1) \mu}{\lambda}\right)\right) \\
&\leq e^{1/\lambda} \cdot \exp\left(\frac{1}{\lambda} \cdot \frac{1}{1 - \exp(-\mu / \lambda)}\right),
\end{align*}
which is at most $e^{\eps/2}$ for our setting of parameters, thereby completing the proof of~\eqref{eq:above-threshold-prob}.

We next prove~\eqref{eq:below-threshold-prob}. Similarly, since $f_{0_{2^\ell}}(\bD) > f_{0_{2^\ell}}(\bD')$, we have $\Pr[\hf_{0_{2^\ell}}(\bD) \leq \tau_\ell] \leq \Pr[\hf_{0_{2^\ell}}(\bD') \leq \tau_\ell]$. Therefore,
\begin{align}
\prod_{\ell \in [\log d]} \left(\prod_{s \in \left(S_{\Ileft^\ell} \circ S_{\Iright^\ell}\right) \setminus S_{I^\ell} \atop s \in \left\{0_{2^\ell}, 10_{2^{\ell - 1}}\right\}} \frac{\Pr[\hf_s(\bD) \leq \tau_\ell]}{\Pr[\hf_s(\bD') \leq \tau_\ell]}\right) 
&\leq \prod_{\ell \in [\log d]} \left(\prod_{s \in \left(S_{\Ileft^\ell} \circ S_{\Iright^\ell}\right) \setminus S_{I^\ell} \atop s = 10_{2^{\ell - 1}}} \frac{\Pr[\hf_s^I(\bD) \leq \tau_\ell]}{\Pr[\hf_s^I(\bD') \leq \tau_\ell]}\right) \nonumber \\
&= \prod_{\ell \in [\log d] \atop 1 0_{2^\ell - 1} \in S^{[0, 2^\ell + 1)}_\ell} \frac{\Pr[\hf_{1 0_{2^\ell - 1}}(\bD) \leq \tau_\ell]}{\Pr[\hf_{1 0_{2^\ell - 1}}(\bD') \leq \tau_\ell]}.
\end{align}
Now, let $L_1$ denote the smallest integer such that $\hf_{10_{2^{L_1} -1}}(\bD') > \tau_{L_1} - \mu$. First, notice that for all $\ell < L_1$ we have $\min\{\hf_{10_{2^{\ell} -1}}(\bD'), \tau_\ell - \mu\} = \tau_\ell - \mu = \min\{\hf_{10_{2^{\ell} -1}}(\bD), \tau_\ell - \mu\}$. This means that
\begin{align*}
\forall \ell < L_1, \frac{\Pr[\hf_{1 0_{2^\ell - 1}}(\bD) \leq \tau_\ell]}{\Pr[\hf_{1 0_{2^\ell - 1}}(\bD') \leq \tau_\ell]} = 1.
\end{align*}
Next, consider $\ell = L_1$. In this case, notice that $\min\{\hf_{1 0_{2^\ell - 1}}(\bD'), \tau_\ell - \mu\} - \min\{\hf_{1 0_{2^\ell - 1}}(\bD), \tau_\ell - \mu\} \leq 1$. As a result, by standard DP property of the Laplace mechanism, we have
\begin{align*}
\frac{\Pr[\hf_{1 0_{2^{L_1} - 1}}(\bD) \leq \tau_\ell]}{\Pr[\hf_{1 0_{2^{L_1} - 1}}(\bD') \leq \tau_\ell]} \leq e^{1/\lambda}.
\end{align*}
Finally, for $\ell > L_1$, we have $f_{10_{2^{\ell} - 1}}(\bD) > \tau_\ell + (\ell - L_1)\mu - 1$. This implies that
\begin{align*}
\frac{\Pr[\hf_{1 0_{2^\ell - 1}}(\bD) \leq \tau_\ell]}{\Pr[\hf_{1 0_{2^\ell - 1}}(\bD') \leq \tau_\ell]}
&\leq \frac{\Pr[(\ell - L_1)\mu + \Lap(\lambda) \leq 0]}{\Pr[(\ell - L_1)\mu - 1 + \Lap(\lambda) \leq 0]} \\
&\leq \exp\left(\frac{1}{\lambda} \cdot \exp\left(\frac{1 - (\ell - L_1)\mu}{\lambda}\right)\right),
\end{align*}
where the last inequality follows from~\cite[Lemma 2.1]{ZhangXX16}.

Combining these inequalities, we have
\begin{align*}
\prod_{\ell \in [\log d]} \left(\prod_{s \in \left(S_{\Ileft^\ell} \circ S_{\Iright^\ell}\right) \setminus S_{I^\ell} \atop s \in \left\{0_{2^\ell}, 10_{2^{\ell - 1}}\right\}} \frac{\Pr[\hf_s^I(\bD) \leq \tau_\ell]}{\Pr[\hf_s^I(\bD') \leq \tau_\ell]}\right) 
&\leq e^{1/\lambda} \cdot \prod_{i=1}^{\infty} \exp\left(\frac{1}{\lambda} \cdot \exp\left(\frac{1 - i \mu}{\lambda}\right)\right) \\
&\leq e^{1/\lambda} \cdot \prod_{i=1}^{\infty} \exp\left(\frac{1}{\lambda} \cdot \exp\left(\frac{-(i - 1) \mu}{\lambda}\right)\right) \\
&\leq e^{1/\lambda} \cdot \exp\left(\frac{1}{\lambda} \cdot \frac{1}{1 - \exp(-\mu / \lambda)}\right),
\end{align*}
which is at most $e^{\eps/2}$ for our setting of parameters, thereby completing the proof of~\eqref{eq:below-threshold-prob}.
\end{proof}

\subsubsection{Utility Analysis}

We next prove the utility analysis of the algorithm, which consists of showing that we discover all heavy hitters and that the expected list size is small. The latter also implies that the expected running time of the algorithm is small, as stated below.

\begin{lemma}[Utility Guarantee of \textsc{PrivHeavyHitter}] \label{lem:priv-tree-util}
Let $\eta, \nu \in (0, 0.1]$. Suppose that $\tau = 0.5\nu n$, $\tau \geq 8 \mu \log d + 8 \lambda \log(d/(\eta \nu))$ and $\mu \geq \lambda \log(16/\nu)$. Then, we have
\begin{itemize}
\item (Heavy Hitters Discovered) W.p. $1 - 0.5\eta$, $L_{[d]}$ contains all $s \in \{0, 1\}^d$ such that $f_s \geq 2\tau$.
\item (Expected List Size) $\E[|L_{[d]}|] \leq 8 / \nu$.
\item (Expected Running Time) the expected running time of the algorithm is $\poly(d/\nu)$.
\end{itemize}
\end{lemma}

\begin{proof}
\begin{itemize}
\item \textbf{(Heavy Hitter Discovered)} Fix any $s$ such that $f_s \geq 2\tau$. We will prove that $\Pr[s \notin L] \leq 0.5 \eta \nu$; since there are at most $1 / \nu$ such $s$'s, a union bound yields the claimed result.

We can bound the desired probability as follows:
\begin{align*}
\Pr[s \notin L_{[d]}] &= \Pr\left[\exists I \in \cI, s|_I \notin L_I\right] \\
&\leq \sum_{I \in \cI} \Pr[s|_I \notin L_I \mid s|_{\Ileft} \in L_{\Ileft} \wedge s|_{\Iright} \in L_{\Iright}] \\
&= \sum_{I \in \cI} \Pr[f^I_{s|_I} + \Lap(\lambda) \leq \tau_\ell] \\
&\leq \sum_{I \in \cI} \Pr[2\tau + \Lap(\lambda) \leq \tau + \mu \log d] \\
&\leq \sum_{I \in \cI} \Pr[0.5\tau + \Lap(\lambda) \leq 0] \\
&\leq \sum_{I \in \cI} \exp(-0.5\tau / \lambda) / 2 \\
&\leq 2d \cdot \exp(-0.5\tau / \lambda)/2 \\
&\leq 0.5 \eta \nu.
\end{align*}

\item \textbf{(Expected List Size)} We will prove by induction on $\ell$ that $\E[|L_I|] \leq 8 / \nu$ for all $\ell \in \{0, \dots, \log d\}$ and $I \in \cI_\ell$. Note that the case $\ell = \log d$ implies the desired bound $\E[|L_{[d]}|]$.

The base case $\ell = 0$ clearly holds since $L_I = \{0, 1\}$.

We next consider any $\ell \in [\log d]$ and assume that the inductive hypothesis holds for $\ell - 1$ and any $I' \in \cI_{\ell - 1}$. Consider any $I \in \cI_\ell$ and let $S := \{s \in \{0, 1\}^{2^\ell} \mid f^I_s > 0.5\tau\}$. Notice that $|S| < n/(0.5\tau) = 4/\nu$. Next, notice that for any $s \notin S$, we have $f^I_s \leq \tau_\ell - \mu$. Therefore, for any $s \in (S_{\Ileft} \circ S_{\Iright}) \setminus S$, we have
\begin{align} \label{eq:drop-prob}
\Pr[s \in L_I] = \Pr[\tau_\ell - \mu + \Lap(\lambda) > \tau_\ell] \leq \exp\left(-\mu/\lambda\right) \leq \nu / 16.
\end{align}
Therefore, we have
\begin{align*}
\E[|L_I|] 
&= \E[|S \cap L_I|] + \E[|L_I \setminus S|] \\
&\leq |S| + \sum_{s \in \{0, 1\}^{2^\ell} \setminus S} \Pr[s \in L_I] \\
&= |S| + \sum_{s \in \{0, 1\}^{2^\ell} \setminus S} \Pr[s \in L_I \mid s \in L_{\Ileft} \circ L_{\Iright}] \Pr[s \in L_{\Ileft} \circ L_{\Iright}] \\
&\overset{\eqref{eq:drop-prob}}{\leq} 4 / \nu + \sum_{s \in \{0, 1\}^{2^\ell} \setminus S} \nu / 16 \cdot \Pr[s \in S_{\Ileft} \circ S_{\Iright}] \\
&\leq 4 / \nu + \nu / 16 \cdot \E[|L_{\Ileft} \circ L_{\Iright}|] \\
&= 4 / \nu + \nu / 16 \cdot \E[|L_{\Ileft}|]\E[|L_{\Iright}|] \\
(\text{Inductive Hypothesis}) &\leq 4 / \nu + \nu / 16 \cdot (8 / \nu)(8 / \nu) \\
&= 8 / \nu,
\end{align*}
which completes the proof of the inductive step.

\item \textbf{(Expected Running Time)} Notice that the expected running time of the algorithm is
\begin{align*}
\poly(d) \cdot \sum_{I \in \cI} \E[|L_{\Ileft}|] \cdot \E[|L_{\Iright}|],
\end{align*}
which is at most $\poly(d/\nu)$ due to the statement shown in the previous item.
\qedhere
\end{itemize}
\end{proof}

\subsubsection{Putting Things Together: Proof of~\Cref{thm:heavy-hitters-main}}

\begin{proof}[Proof of \Cref{thm:heavy-hitters-main}]
We first run \textsc{PrivHeavyHitter} with parameter $\lambda = 3/\eps, \mu = \lambda \log(16/\nu), \tau = 0.5\nu n$. Let $L$ denote its output. We then let $\hf_x = f_x + \Lap(2/\eps)$ for every $x \in L$ and output $(L, \hf)$.

From \Cref{lem:priv-tree-privacy}, \textsc{PrivHeavyHitter} with specified parameters is $0.5\eps$-DP. Furthermore, the second step is simply the $0.5\eps$-DP Laplace mechanism. Thus, the entire algorithm is $\eps$-DP.

We next argue its utility assuming that $$n \geq \frac{16 \mu \log d + 16 \lambda \log(d/(\eta \nu))}{\nu} + \frac{100 \log(1/(\eta\nu))}{\nu \eps} = O\left(\frac{1}{\eps\nu} \cdot \log(d/\eta)\log(1/\nu)\right).$$ By our setting of parameters, we may apply \Cref{lem:priv-tree-util} to conclude that w.p. $1 - 0.5\eta$ all desired heavy hitters belong to the list $L_{[d]}$ and that $\E[|L_{[d]}|] \leq 8/\nu$. The latter together with Markov inequality further implies that w.p. $1 - 0.25\eta$ we have $|L_{[d]}| \leq 32 / \nu$. When this occurs, we may use a union bound to conclude that
\begin{align*}
\Pr[\forall x \in L, |f_x - \hf_x| \leq \nu n] &\geq 1 - \sum_{x \in L} \Pr[|f_x - \hf_x| \leq \nu n] \\
&\geq 1 - \sum_{x \in L} \Pr[|\Lap(2/\eps)| \leq \nu n] \\
&\geq 1 - (32/\nu) \cdot \exp(-0.05 \nu n \eps),
\end{align*}
which is at least $1 - 0.25\eta$ since $n \geq \frac{100 \log(1/(\eta\nu))}{\nu \eps}$. Using a union bound again, we can conclude that our algorithm solves the histogram problem to within an error $\nu$ with probability at least $1 - \eta$.
\end{proof}

\subsection{Histogram: Lower Bound}

\begin{proof}[Proof of \Cref{thm:heavy-hitters-lb}]
Let $\bD^{(1)}, \dots, \bD^{(d)}$ be datasets such that $\bD^{(i)}$ contains $\lceil 3 \nu n\rceil$ copies of the one-hot vector $\bone_{i}$ and $n - \lceil 3 \nu n\rceil$ copies of the all zeros vector. Furthermore, let $\bD^{(0)}$ denote the dataset of $n$ all-zero vectors.

Let $S^{(i)}$ denote the set of solutions of error at most $\nu$ for $\bD^{(i)}$. For any distinct $i, j \in [d]$, we claim that $S^{(i)}$ and $S^{(j)}$ are disjoint. This holds because any $\{\hf_x\}_x \in S^{(i)}$ must satisfy $\hf_{\bone_i} \geq 2\nu n$ but any $(L, \hf) \in S^{(j)}$ must satisfy $\hf_{\bone_i} \leq \nu n$; these two conditions can hold simultaneously.

Observe also that $\bD^{(i)}$ and $\bD^{(0)}$ are $\lceil 2 \nu n \rceil$-neighbor under the partial DP notion. Therefore, for any algorithm $\cA$ that solves the heavy hitter problem with probability $0.1$, we have
\begin{align*}
1 \geq \Pr\left[\cA(\bD^{(0)}) \in \bigcup_{i \in [d]} S^{(i)}\right]
&= \sum_{i \in [d]} \Pr[\cA(\bD^{(0)}) \in S^{(i)}] \\
&\geq \sum_{i \in [d]} e^{-\eps \lceil 2 \nu n \rceil} \Pr[\cA(\bD^{(i)}) \in S^{(i)}] \\
&\geq d \cdot e^{-\eps \lceil 2 \nu n \rceil} \cdot 0.1 \\
(\text{From } d \geq 10 e^{1.1\eps}) &\geq d^{0.1} e^{-2 \eps \nu n},
\end{align*}
which indeed implies that $n \geq \frac{0.05}{\eps \nu} \log d$ as desired.
\end{proof}

\subsection{PAC Learning with Partial DP}

We next consider the PAC learning setting where there is an unknown distribution $\cD$ on $\{0, 1\}^d \times \{0, 1\}$ and the learner receives $n$ i.i.d. samples $(x_1, y_1), \dots, (x_n, y_n)$ drawn from $\cD$. The goal is to output a hypothesis $h: \{0, 1\}^d \to \{0, 1\}$ that minimizes the population error $\err(h; \cD) := \Pr_{(x, y) \sim \cD}[h(x) \ne y]$. An algorithm is said to be a \emph{PAC learner for a hypothesis class $\cH$ with error at most $\alpha$} if and only if, assuming that $\cD$ is realizable with respect to a hypothesis from $\cH$, with probability at least $0.9$ it outputs a hypothesis $h'$ such that $\err(h'; \cD) \leq \alpha$. An algorithm is said to be \emph{proper} if its output $h$ belongs to the hypothesis class $\cH$.

We note that, for \PDP, we view the concatenation of $x_i$ and $y_i$ as the input to user's $i$, i.e., a label is treated as another attribute.

\subsubsection{Learning Point Functions}

A \emph{point function} is a function $h_u$ where $\point_u(z) := \bone[u = z]$. The class of point functions, $\cH_{\point}$, is defined as $\{\point_u \mid u \in \{0, 1\}^d\}$. Learning point functions is well understood in the standard $\eps$-DP setting: the sample complexity is $\Theta_{\alpha}(d / \eps)$ for proper DP learning and $\Theta_{\alpha}(1/\eps)$ for improper DP learning~\cite{BeimelBKN14}. Below we show that the former can improved to $O_{\alpha}(\log d / \eps)$ for partial DP.

Our partial DP algorithm for learning point function uses our (succint) histogram algorithm to find the set of heavy hitters among $x_iy_i$'s. If a heavy hitter has $y = 1$, then output $\point_x$ for the corresponding $x$. Otherwise, if all heavy hitters have $y = 0$, then output any point function $\point_x$ whose $x$ does not appear in the list.

\begin{theorem} \label{thm:point-func}
There exists an $\eps$-\PDP proper PAC learner with error at most $\alpha$ and with sample complexity $O\left(\frac{1}{\alpha\eps} \log d \log(1/\alpha)\right)$ for point functions.
\end{theorem}

\begin{proof}[Proof of \Cref{thm:point-func}]
Our algorithm works as follows. First, we let $z_i$ be the concatenation of $x_i$ and $y_i$ and run the $\eps$-\PDP heavy hitter algorithm from~\Cref{thm:heavy-hitters-main} with $\nu = 0.2\alpha, \eta = 0.01$ to obtain an approximate frequency $\hf$ for $z_1, \dots, z_n \in \{0, 1\}^d$. We then select the output hypothesis as follows:
\begin{itemize}
\item We attempt to find $x^1 = 1$ and $\hf(x^1 1) > \nu n$ (with ties broken arbitrarily). If such an $x^1$ is found, then we output $h_{x^1}$.
\item Otherwise, if no such $z^*$ is found, we attempt to find $x^0$ such that $\hf(x^0 0) \leq \nu n$ (with ties broken arbitrarily). If such an $x^0$ is found, then we output $h_{x^0}$.
\item If neither $z^*$ or $x^0$ is found, then we output an arbitrary hypothesis from $\cH_{\text{point}}$.
\end{itemize}

It is immediate that the algorithm is $\eps$-\PDP. We next analyze its error guarantee assuming that $n \geq \frac{C}{\alpha\eps} \log d \log(1/\alpha)$ where $C$ is a sufficiently large constant. A standard generalization bound implies that with probability at least $0.99$, it holds that $|\err(h; S) - \err(h; \cD)| \leq 0.1\alpha$ for all $h \in \cH_{\text{point}}$. Recall also from~\Cref{thm:heavy-hitters-main} that with probability at least $0.99$, the following holds: $|\hf_z - f_z| \leq 0.1\nu n$ for all $z \in \{0, 1\}^{d + 1}$. We will henceforth assume that this event occurs.

Since $\cD$ is realizable by some hypothesis $h_{x^*} \in \cH$, $S$ is also consistent with this hypothesis. Consider two cases based on whether we found $x^1$.
\begin{itemize}
\item If we found $x^1$, then it must be equal to $x^*$ because for all $x \ne x^*$ we have that $\hf_{x \circ 1} \leq f_{x \circ 1} + \nu n = \nu n$. Therefore, $\err(h_{x^1}; \cD) = 0$ as desired.
\item If we do not find $x^1$, then we have that $\hf_{x^* 1} < 2 \nu n$. Furthermore, since $\hf_{x^* 0} \leq f_{x^* 0} + \nu n = \nu n$, we are guaranteed to find $x^0$. Thus, we have that
\begin{align*}
\err(h_{x^0}; \cD) \leq 0.1\alpha + \err(h_{x^0}; S) = 0.1\alpha + \nu + f_{x^0 0} / n + f_{x 1} / n &\leq 0.1 \alpha + 2\nu + \hf_{x^0 0} / n + \hf_{x 1} / n \\
&\leq 2.1 \nu + \nu + \nu = 4.1\nu < \alpha.
\end{align*}
\end{itemize}
Hence, in both cases, the population error of the output hypothesis is at most $\alpha$ as desired.
\end{proof}

\subsubsection{Learning Threshold Functions}

Let the elements in $\{0, 1\}^d$ be ordered lexicographically. %
A \emph{threshold function} $\Thre_z$ is defined as $\Thre_z(x) := \bone[x \geq z]$. The class of threshold functions is $\cH_{\Thre} := \{\Thre_z \mid z \in \{0, 1\}^d\}$. Learning threshold functions is among the most well-studied problem in the DP literature~\cite{BeimelNS19,BeimelNS16,BunNSV15,FeldmanX15,AlonLMM19,KaplanLMNS20}. In the standard $\eps$-DP setting, it is known that the sample complexity is $\Theta_\alpha(d / \eps)$. Here we will show that this can be reduced to $\Theta_\alpha(\log d / \eps)$ under $\eps$-\PDP.

Our algorithm uses the histogram algorithm to find a longest ``polarizing'' prefix, i.e., a prefix $p$ for which there are sufficiently many input examples of the form $(p s, 0)$ and $(p s, 1)$. %
Intuitively, this is the longest prefix for which we are still unsure what to label. Once such a prefix is found, it is now easy to find a low-error threshold function as we only have to figure out what to label those with prefix $p0$ and $p1$.

We start by describing the algorithm for finding a longest polarizing prefix. For a prefix $p$, we write $f_{p*0}$ (resp. $f_{p*1}$) to denote the number of input samples of the form $(ps, 0)$ (resp. $(ps, 1)$). We define $\pol(p) := \min\{f_{p*0}, f_{p*1}\} / n$, and say that a prefix $p$ is \emph{$\gamma$-polarizing} if $\pol(p) \geq \gamma$. Our subroutine has the following guarantee.

\begin{lemma} \label{lem:polarizing-prefix}
There exists an $\eps$-\PDP algorithm that, given a set $S$ of $n \geq \tilde{O}\left(\frac{1}{\eps \gamma} \log d\right)$ labeled examples where $\pol(\perp) \geq 2\gamma$, with probability $0.99$ finds a prefix $p$ that is $\gamma$-polarizing and no prefix longer than $p$ is $2\gamma$-polarizing.
\end{lemma}

Here we use $\perp$ to denote the empty string.

\begin{proof}
For $\ell \in \{0, \dots, \log d\}$, let $\max\pol_\ell$ denote $\max_{p \in \{0, 1\}^\ell} \pol(p)$. The algorithm starts by using a binary search to find $\ell^*$ such that $\widetilde{\max\pol}_\ell = \max\pol_\ell + \Lap(2 \log d/(\eps n)) \geq 1.5\gamma$; if no such $\ell^*$ exists, simply let $\ell^* = 0$. Then, we use our $0.5\eps$-\PDP histogram algorithm from \Cref{thm:heavy-hitters-main} on the set $S_{\ell^*} := \{x|_{[\ell^*]} y \mid (x, y) \in S\}$ with $\nu = 0.1\gamma$. Finally, we attempt to find $p$ such that $\hf_{p0}, \hf_{p1} \geq 1.5\gamma$ and output such a $p$ if found (tie broken arbitrarily).

The privacy guarantee of the algorithm follows directly from that of the Laplace mechanism, the histogram algorithm, and basic composition. 

We will prove the guarantee assuming that $n \geq C \cdot \left(\frac{1}{\eps \gamma} \log d \log(1/\gamma) \log \log d\right)$, where $C$ is a sufficiently large constant. Under this assumption, w.p. 0.995 we have $|\widetilde{\max\pol}_\ell - \max\pol_\ell| < 0.1 \gamma$ for all $\ell$ used during the binary search. This indeed implies that there is no $2\gamma$-polarizing prefix of length $\ell^* + 1$ and that there exists an $1.4\gamma$-polarizing prefix of length $\ell^*$. Now, using the guarantee of~\Cref{thm:heavy-hitters-main}, we can conclude that we find a $\gamma$-polarizing prefix with probability at least $0.99$ as desired.
\end{proof}

We can now use the above algorithm to learn threshold functions.

\begin{theorem} \label{thm:learning-thresholds}
There exists an $\eps$-\PDP proper PAC learner with error at most $\alpha$ and with sample complexity $\tilde{O}\left(\frac{1}{\alpha\eps} \log d\right)$ for threshold functions.
\end{theorem}

\begin{proof}
Our algorithm works as follows:
\begin{itemize}
\item First, we privately compute the fraction of 1 labels, i.e., $a = \E_{(x, y) \sim S}[y] + \Lap(3/(\eps n))$.
\item If $a \leq 0.5\alpha$, then output $\Thre_{0_d}$.
\item Otherwise, if $a \geq 1 - 0.5\alpha$, then output $\Thre_{1_d}$.
\item Otherwise, if $a \in (0.5\alpha, 1 - 0.5\alpha)$, run the $(\eps/3)$-\PDP algorithm from~\Cref{lem:polarizing-prefix} with $\gamma = 0.1\alpha$ to find a prefix $p \in \{0, 1\}^\ell$. Finally, use an $(\eps/3)$-DP exponential mechanism to select among the hypotheses $\Thre_{p0_{d}}, \Thre_{p10_{d - \ell - 1}}, \Thre_{p1_{d}}$.
\end{itemize}

The privacy guarantee of the algorithm follows from that of Laplace mechanism and basic composition. We next argue its accuracy assuming that $n \geq C \cdot \left(\frac{1}{\eps \alpha} \log d \log(1/\alpha) \log \log d\right)$ for some sufficiently large constant $C$. For this, we will only show that the empirical error $\err(h; S)$ of the output hypothesis $h$ is at most $0.9\alpha$ w.p. 0.95. A standard generalization bound then implies that $\err(h; \cD)$ is at most $\alpha$ w.p. 0.9.

To bound $\err(h; S)$, first notice that, w.p. 0.99, the Laplace noise added in the first step is at most $0.1\alpha$. Therefore, if we output either $\Thre_{0_d}$ or $\Thre_{1_d}$, then their errors are at most $0.6\alpha$. Furthermore, if we proceed to the last step, then we must have $\pol(\lambda) \geq 0.4\alpha$. Therefore, \Cref{lem:polarizing-prefix} ensures that with probability 0.99, we indeed finds $p$ that is $0.1\alpha$-polarizing and that neither $p0$ nor $p1$ is $0.2\alpha$-polarizing. Since $p$ is $0.1\alpha$-polarizing, it must be that the underlying hypothesis is $\Thre_{p s}$ for some $s \in \{0, 1\}^d$. Since neither $p0$ nor $p1$ is $0.2\alpha$-polarizing, at least one of $\Thre_{p0_{d}}, \Thre_{p10_{d - \ell - 1}}, \Thre_{p1_{d}}$ must have empirical error at most $0.2\alpha$. As a result, with probability 0.99, the exponential mechanism will indeed pick a hypothesis $h$ such that $\err(h; S) \leq 0.9\alpha$ as desired.
\end{proof}

\subsection{Discrete Distribution Estimation Under $\ell_2^2$ Error}

Let $\Delta_{\{0, 1\}^d} := \{\phi \in [0, 1]^{\{0, 1\}^d} \mid \sum_{x \in \{0, 1\}^d} \theta_x = 1\}$, where $\theta \in \Delta_{\{0, 1\}^d}$ is also viewed as a probability distribution over $\{0, 1\}^d$. In discrete distribution estimation, we are given samples $x_1, \dots, x_n \sim \theta$ for some unknown $\theta \in \Delta_{\{0, 1\}^d}$. The goal is to output $\theta^{priv} \in [0, 1]^d$ that minimizes the $\ell_2^2$ (aka squared) error: $\|\theta^{priv} - \theta\|_2^2 = \sum_{x \in \{0, 1\}^d} (\theta^{priv}_x - \theta_x)^2$. The problem is well studied in the local (standard) DP setting (e.g.,~\cite{DuchiWJ13,KairouzBR16,YeB18}). We are not aware of works that studies it in the central DP setting before but it is not hard to see (using similar lower bound as in histogram~\cite{HardtT10}) that to get a constant say 0.1 error $n$ has to be at least $\Omega\left(d/\eps\right)$. On the other hand, below we show that this is possible in partial DP even when $n = \tilde{O}(\log d / \eps)$.

Our algorithm uses the histogram algorithm to discover all $x$ with $\theta_x \geq \nu$. For these $x$'s, we then use the Laplace mechanism to estimate $\theta^{priv}_x$. For the remaining $x$'s, we output $\theta^{priv}_x = 0$; since $\theta_x < \nu$, these zero outputs only contribute to at most $\nu$ in the $\ell_2^2$ error.

\begin{theorem} \label{thm:dist-est-l2}
There exists an $\eps$-\PDP algorithm for discrete distribution whose $\ell_2^2$ error is at most $\tilde{O}\left(\frac{\log d}{\eps n} + \frac{1}{n}\right)$ with probability at least $0.9$.
\end{theorem}

\begin{proof}[Proof of \Cref{thm:dist-est-l2}]
We run the $0.5\eps$-\PDP heavy hitter algorithm from~\Cref{thm:heavy-hitters-main} with $\nu = C \cdot \frac{\log d}{\eps n} \cdot \frac{1}{\log(\eps n)}, \eta = 0.01$ where $C$ is approximate frequencies $(\hf_x)_{x \in \{0, 1\}^d}$ for w.r.t. the input $x_1, \dots, x_n \in \{0, 1\}^d$. We then let
\begin{align*}
\theta^{priv}_x := 
\begin{cases}
(f_x + \Lap(2/\eps)) / n &\text{ if } \hf_x \geq 2\nu n, \\
0 &\text{ otherwise.}
\end{cases}
\end{align*}

It is obvious that the algorithm is $\eps$-\PDP. Next, we analyze its error guarantee. First, notice that
\begin{align*}
\|\theta^{priv} - \theta\|_2^2 \leq 2\|\theta^{priv} - f / n\|_2^2 + 2\|f / n - \theta\|_2^2.
\end{align*}
Hence, it suffices to show that $\Pr\left[\|\theta^{priv} - f/n\|_2^2 \leq \tilde{O}\left(\frac{\log d}{\eps n}\right)\right] \geq 0.95$ and $\Pr[\|f/n - \theta\|_2^2 \leq O(1/n)] \geq 0.95$. We start with the former. For this, recall from~\Cref{thm:heavy-hitters-main} that with probability 0.99, the following holds: $|\hf_x - f_x| \leq \nu n$ for all $x \in \{0, 1\}^d$. We will henceforth assume this event occurs. Notice that we may write
\begin{align*}
\E[\|\theta^{priv} - f/n\|_2^2] &= \sum_{x \in \{0, 1\}^d \mid \hf_x \geq 2\nu n} \E[(\theta^{priv}_x - f_x / n)^2] + \sum_{x \in \{0, 1\}^d \mid \hf_x < 2\nu n} \E[(f_x / n)^2] \\
&\leq \sum_{x \in \{0, 1\}^d \mid \hf_x \geq 2\nu n} O\left(\frac{1}{\eps^2n^2}\right) + \sum_{x \in \{0, 1\}^d \mid \hf_x < 2\nu n} \E[(3\nu)(f_x / n)] \\
&\leq O\left(\frac{1}{\eps^2 n^2 \nu}\right) + O(\nu) \\
&= \tilde{O}\left(\frac{\log d}{\eps n}\right).
\end{align*}
Applying the Markov inequality yields the desired probabilistic bound.

For the latter, we have
\begin{align*}
\E[\|f/n - \theta\|_2^2] &= \sum_{x \in \{0, 1\}^d} \E[\left(f_x / n - \cdot \theta_x\right)^2] \\
&= \sum_{x \in \{0, 1\}^d} \theta_x(1 - \theta_x) / n \\
&\leq \sum_{x \in \{0, 1\}^d} \theta_x / n = 1/n.
\end{align*}
Applying the Markov inequality yields the desired probabilistic bound.
\end{proof}

\end{document}